\documentclass[11pt]{article}

\usepackage[margin=1in]{geometry}
\usepackage{amsmath,amssymb,amsfonts,amsthm,mathtools}
\usepackage{bm}
\usepackage{booktabs}
\usepackage{tabularx}
\usepackage{array}
\usepackage{graphicx}
\usepackage{caption}
\usepackage{subcaption}
\usepackage{authblk}
\usepackage{enumitem}
\usepackage{hyperref}
\usepackage[nameinlink,capitalize]{cleveref}
\usepackage{setspace}
\usepackage{microtype}
\usepackage{color}
\usepackage{placeins}
\usepackage{titlesec}

\hypersetup{
    colorlinks=true,
    linkcolor=blue,
    citecolor=blue,
    urlcolor=blue
}

\newtheorem{proposition}{Proposition}

\newcommand{\E}{\mathbb{E}}
\newcommand{\Var}{\operatorname{Var}}

\newcommand{\Prob}{\mathbb{P}}
\newcommand{\dd}{\,\mathrm{d}}
\newcommand{\e}{\mathrm{e}}
\newcommand{\ind}{\mathbf{1}}
\newcommand{\Treb}{T_{\rm reb}}
\newcommand{\Vdet}{V_{\rm det}}
\newcommand{\td}{\tau_{\rm d}}
\newcommand{\tecl}{\tau_{\rm e}}
\newcommand{\tdelay}{\tau_{\rm det}}
\newcommand{\tw}{t_w}
\newcommand{\Lcal}{\mathcal{L}}

\title{\bfseries Stochastic First-Passage Theory of HIV Viral Rebound Following Latent Reservoir Reactivation}

\author[1]{Mesfin Asfaw Taye\thanks{Email: tayem@wlac.edu}}
\affil[1]{West Los Angeles College, Science Division, 9000 Overland Ave, Culver City, CA 90230, USA}

\date{\today}

\begin{document}
\maketitle


\begin{abstract}
In our earlier work, we modeled the stochastic initiation of HIV rebound by treating latent-cell reactivation as a Poisson-driven process during antiretroviral-therapy (ART) washout, immune modulation, and therapeutic perturbation~\cite{Taye2025CM}. That framework characterized activation survival, cumulative hazards, waiting-time laws, and expected viral-load trajectories. However, the endpoint observed in analytical treatment interruption (ATI) studies is not the hidden time of first successful reactivation. It is the first time at which plasma virus exceeds an assay-defined detection threshold. Here we reformulate post-treatment rebound as a stochastic first-passage problem, with $T_{\rm reb}=\inf\{t\ge t_w:V(t)\ge V_{\rm det}\}$. Successful reactivation events arrive with a time-dependent intensity, and each event seeds an exponentially expanding viral lineage. The total plasma viral load is therefore a Poisson shot-noise process, and rebound corresponds to its first threshold crossing. In the rare-reactivation regime, this crossing is dominated by the earliest successful lineage. Rebound timing then separates into two components: a stochastic waiting time for reservoir reactivation and a deterministic growth delay to detectability. This separation gives a shifted-hazard survival law and yields closed-form rebound-time distributions for constant activation, ART-washout-dependent activation, immune-periodic activation, Cox-process activation, and heterogeneous-reservoir activation. The same formulation also provides a likelihood suitable for the interval-censored sampling structure of ATI trials. A central prediction is that observed rebound times depend logarithmically on the detection threshold. Threshold-resolved medians from a large ATI meta-analysis, approximately $16$, $21$, and $32$ days at thresholds of $50$, $400$, and $10{,}000$ copies/mL, support this dependence. They imply an effective net growth rate of about $0.33\,{\rm day}^{-1}$, below the nominal maximal early-infection value. This calibration explains why maximal-growth parameters underpredict observed rebound times and links latent-reservoir reactivation directly to the measurable ATI rebound endpoint.
\end{abstract}


\section{Introduction}

Combination antiretroviral therapy (ART) has transformed HIV-1 infection from a rapidly fatal disease into a clinically manageable chronic condition. ART suppresses viral replication, restores immune function, and delays disease progression~\cite{Perelson1993,Perelson1996,Ho1995,Wei1995}. ART, however, is suppressive rather than curative. A long-lived reservoir of latently infected quiescent CD4$^+$ T cells is established early and persists through transcriptional silence, immune evasion, homeostatic proliferation, and clonal expansion~\cite{Chun1997,Siliciano2003,Murray2016,Li2020}. Because this reservoir is largely inaccessible to therapies that target active viral replication, treatment interruption can reseed productive infection and produce plasma viral rebound within days to weeks. Accurate prediction of rebound time is therefore central to the design and evaluation of curative strategies.

Viral rebound is not a deterministic resumption of viral replication, but a stochastic process initiated by rare and heterogeneous reservoir-reactivation events within a small reservoir. A latent cell must reactivate, the emerging lineage must avoid early stochastic extinction, and the established lineage must expand to a detectable plasma level. These steps are shaped by reservoir size, activation state, immune surveillance, residual drug activity, anatomical compartmentalization, and intrinsic small-number fluctuations~\cite{Hill2014,Hill2018,pinkevych2015latency,conway2019rebound}. Deterministic ordinary differential equation models remain essential for population-level virus--cell dynamics~\cite{NowakBangham1996,Perelson2002,Wodarz2002}. Yet analytical treatment interruption (ATI) studies require probabilistic quantities: rebound hazards, waiting-time distributions, and the probability that an individual remains below a specified viral-load threshold at a given time.

The sampling design of ATI studies makes this distinction central to the analysis. The microscopic time of successful reservoir reactivation is not measured directly. Plasma HIV RNA is sampled at discrete visits, and rebound is assigned to the first measurement above an assay-defined threshold. Thus, the clinical rebound endpoint is not the activation time itself, but an interval-censored threshold-crossing time. This distinction is supported by threshold-resolved ATI data. In an individual-participant meta-analysis of 24 prospective ATI studies, Gunst et al.\ reported median times of approximately $16$, $21$, and $32$ days to exceed $50$, $400$, and $10{,}000$ HIV RNA copies/mL, respectively~\cite{gunst2025ati}. Such systematic shifts with $V_{\rm det}$ are naturally explained by stochastic reactivation followed by a growth delay to detectability. They cannot be captured by a theory that identifies rebound with the first activation event alone. Additional evidence points in the same direction: successful initiation appears intermittent, with reactivation events occurring on average every five to eight days after interruption~\cite{pinkevych2015latency}; genetically barcoded SIV studies show that rebound can be seeded by multiple independent lineages~\cite{fennessey2017barcoded,VanDorp2020}; and reservoir measures, including on-ART HIV expression and cell-associated RNA, correlate with rebound timing~\cite{li2016reservoir,pasternak2020carna,sneller2020kinetics}. These observations motivate a stochastic theory that includes activation, establishment, expansion, and delayed detectability~\cite{conway2019rebound}.

In our earlier work, we developed a stochastic model of HIV reactivation under ART washout and immune fluctuations, treating latency reversal as a Poisson-driven process with time-dependent intensity, pharmacokinetic decay, immune variability, and therapeutic perturbation protocols~\cite{Taye2025CM}. That framework described the hidden activation process through cumulative hazards, activation-survival probabilities, waiting-time laws, and ensemble-averaged viral-load trajectories. It did not directly describe the endpoint measured in ATI studies. Activation survival gives the time to successful reactivation, while the mean viral load gives an ensemble average. Neither quantity is equivalent to the first time at which plasma viral load in an individual participant crosses a prescribed detection threshold. This distinction is especially important when reactivation is rare, because the mean viral load may be dominated by a small number of early expanding trajectories and may therefore misrepresent the typical rebound time.

In this work, we close this gap by formulating post-treatment HIV rebound as a stochastic first-passage problem. The rebound time is defined as the first crossing of a detection threshold by an individual, randomly evolving viral load. Successful reactivation events arrive with intensity $\lambda(t)$, and each event seeds a lineage that grows after an eclipse phase. The resulting viral load is a Poisson shot-noise process, and rebound is the first threshold crossing of this process. In the rare-reactivation regime, the crossing is dominated by the earliest successful expanding lineage. Rebound timing then separates into a stochastic waiting time for successful reservoir reactivation and a deterministic growth delay to detectability. For a constant successful-reactivation rate, this gives the compact approximation
\begin{equation}
\E[T_{\rm reb}]
\approx
\tw+\tau_e+\frac{1}{\lambda}
+\frac{1}{r}\log\!\left(\frac{V_{\rm det}}{v_0}\right),
\label{eq:constant_mean_intro}
\end{equation}
where $\tw$ is the washout time, $\tau_e$ is the eclipse delay, $\lambda$ is the successful-reactivation rate, $r$ is the net viral growth rate, $v_0$ is the founder contribution to plasma virus, and $V_{\rm det}$ is the assay threshold. The logarithmic dependence on $V_{\rm det}$ explains why increasing the detection threshold shifts rebound times additively rather than multiplicatively. It also separates the principal biological mechanisms by which interventions may delay rebound: reducing successful reactivation, lowering establishment probability, or slowing post-reactivation viral growth.

The contribution of this paper is threefold. First, it recasts post-treatment HIV rebound as a first-passage problem for a stochastic viral-load process rather than as activation survival alone. Second, it derives shifted-hazard rebound-time laws for general time-dependent successful-reactivation intensities, with explicit forms for constant activation, ART-washout-dependent activation, immune-periodic activation, Cox-process activation, and heterogeneous reservoir classes. Third, it connects the theory directly to ATI data by incorporating detection thresholds, interval censoring, and physiologically interpretable parameters. The threshold-resolved medians reported by Gunst et al.\ are shown to be consistent with an effective net growth rate below the maximal early-infection value. The aim is not to replace mechanistic virus--cell models, but to provide an analytically transparent theory of the quantity ATI studies actually measure: the first crossing of a viral-load detection threshold.

The remainder of the paper is organized as follows. Section~\ref{sec:old_new} distinguishes first-passage rebound from activation-survival and mean-load descriptions. Section~\ref{sec:model} develops the successful-reactivation and viral-expansion framework. Section~\ref{sec:exact_fp} establishes the exact first-passage relation and the single-founder approximation. Sections~\ref{sec:constant}--\ref{sec:cox} derive predictions for constant, washout-dependent, periodic, and stochastic activation profiles, while Section~\ref{sec:heterogeneous} treats competing reservoir classes. Section~\ref{sec:validation_inference_interpretation} connects the model to threshold-dependent ATI observations, interval-censored likelihood inference, and biological interpretation. Detailed derivations are collected in the appendices.

\section{First-Passage Formulation and Empirical Basis}
\label{sec:old_new}

The central quantity analyzed in this work is the first time at which an individual stochastic viral-load trajectory becomes clinically detectable. We define the rebound time as 
\begin{equation}
\Treb=\inf\{t\ge \tw:V(t)\ge \Vdet\},
\label{eq:Treb_firstpassage}
\end{equation}
where \(\tw\) denotes treatment interruption or effective washout, \(V(t)\) is the total viral load generated by successful reactivation events, and \(\Vdet\) is the assay threshold used to define rebound. Accordingly, rebound cannot be identified either with the unobserved time of first latent-cell reactivation or with the deterministic threshold crossing of an ensemble-averaged viral load. It is defined as a first-passage event: the first time a single stochastic viral-load trajectory reaches a prescribed detection boundary. Figure~\ref{fig:first_passage_schematic} illustrates this formulation and the interval-censored sampling structure of ATI data.

\begin{figure}[tbp]
    \centering
    \includegraphics[width=\linewidth]{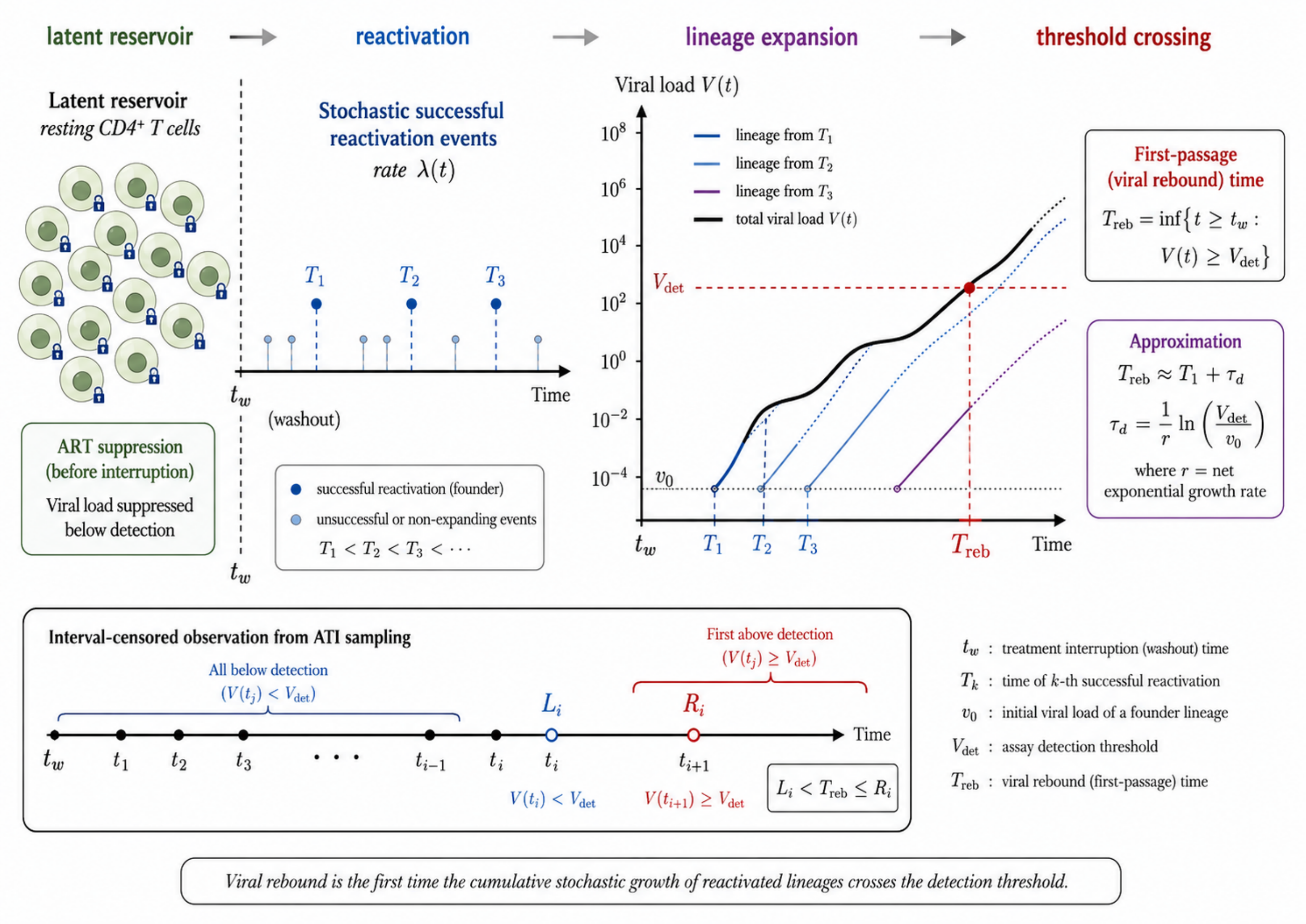}
    \caption{
    First-passage formulation of HIV viral rebound after ART interruption. Before interruption, the latent reservoir persists under ART suppression and plasma virus remains below the detection threshold. After washout at \(t_w\), successful reactivation events occur with time-dependent intensity \(\lambda(t)\). Each event at time \(T_i\) seeds a lineage contributing \(v_0 e^{r(t-T_i)}\), so the total viral load is \(V(t)=v_0\sum_{T_i\le t}e^{r(t-T_i)}\). Rebound is the first crossing \(T_{\rm reb}=\inf\{t\ge t_w:V(t)\ge V_{\rm det}\}\). The lower panel indicates the interval-censored structure of ATI sampling, where \(L_i<T_{{\rm reb},i}\le R_i\), and the single-founder approximation \(T_{\rm reb}\approx T_1+\tau_d\), with \(\tau_d=r^{-1}\ln(V_{\rm det}/v_0)\).
    }
    \label{fig:first_passage_schematic}
\end{figure}

To distinguish observed rebound from the hidden reactivation process, let \(N(t)\) denote the counting process of successful established reactivation events after washout, with time-dependent intensity \(\lambda(t)\). For a given probability distribution,
\[
P(n,t)=\Prob\{N(t)-N(\tw)=n\},
\]
the corresponding inhomogeneous Poisson dynamics is 
\begin{equation}
\frac{\dd P(n,t)}{\dd t}
=
\lambda(t)\left[P(n-1,t)-P(n,t)\right],
\qquad n\ge1,
\label{eq:master_intro}
\end{equation}
with \(\dd P(0,t)/\dd t=-\lambda(t)P(0,t)\) and \(P(n,\tw)=\delta_{n,0}\). Defining the cumulative successful-reactivation hazard
\begin{equation}
\Lambda(t)=\int_{\tw}^{t}\lambda(s)\dd s,
\label{eq:activation_hazard_section}
\end{equation}
one obtains
\begin{equation}
P(n,t)=\frac{\Lambda(t)^n}{n!}\e^{-\Lambda(t)},
\label{eq:poisson_solution_intro}
\end{equation}
and hence the activation-survival probability
\begin{equation}
S_{\rm act}(t)
=
\Prob\{N(t)-N(\tw)=0\}
=
\e^{-\Lambda(t)}.
\label{eq:Sact_section}
\end{equation}
This survival probability belongs to the hidden reactivation process, not to the clinical rebound endpoint. It quantifies the absence of successful reactivation, rather than the probability that plasma viremia remains below a prescribed detection threshold.

A second natural but distinct quantity is the ensemble-mean viral load. If a successful event at time \(s\) contributes \(v_0\) and subsequently grows with net rate \(r=g-c\), then
\begin{equation}
\bar V(t)=\E[V(t)]
=
v_0\int_{\tw}^{t}\lambda(s)\e^{r(t-s)}\dd s .
\label{eq:mean_load_section}
\end{equation}
This average is analytically useful, but it is not the observed rebound time in an ATI participant, who follows one stochastic trajectory. The clinically relevant survival function is instead
\begin{equation}
S_{\rm reb}(t)
=
\Prob(\Treb>t)
=
\Prob\left\{
\sup_{\tw\le u\le t}V(u)<\Vdet
\right\}.
\label{eq:Sreb_sup_section}
\end{equation}
When \(r>0\) and all established lineages grow monotonically, \(V(t)\) is nondecreasing, so the supremum is attained at \(t\) and
\begin{equation}
S_{\rm reb}(t)=\Prob\{V(t)<\Vdet\}.
\label{eq:Sreb_monotone_section}
\end{equation}
Thus, the rebound problem reduces to the distributional first-passage problem of a stochastic viral load.

The three quantities
\begin{equation}
S_{\rm act}(t),\qquad \bar V(t),\qquad S_{\rm reb}(t)
\end{equation}
are therefore generally distinct:
\begin{equation}
S_{\rm reb}(t)\ne S_{\rm act}(t),
\qquad
\Treb \ne T_1,
\qquad
\Treb \ne \inf\{t:\E[V(t)]\ge \Vdet\}.
\label{eq:three_distinctions}
\end{equation}
The first inequality reflects the growth delay between successful reactivation and detectability. The second states that the first established lineage is not observed at activation but only after sufficient expansion. The third states that the threshold crossing of an ensemble mean is not the first passage of a random trajectory. This last distinction is especially important in the rare-reactivation regime, where the ensemble mean may be dominated by a minority of early high-growth trajectories while most individuals remain below detection.

The first-passage formulation also matches the observational structure of ATI studies. Plasma HIV RNA is sampled at discrete times after ART interruption, and rebound is recorded only when the measured viral load first exceeds a specified threshold. If \(L_i\) is the last sampling time at which participant \(i\) remains below threshold and \(R_i\) is the first above-threshold sampling time, then
\begin{equation}
L_i<T_{{\rm reb},i}\le R_i .
\label{eq:interval_censoring_motiv}
\end{equation}
Thus, rebound times are naturally interval-censored threshold-crossing events rather than exactly observed activation times.

The threshold dependence of the Gunst et al.\ medians~\cite{gunst2025ati} is a direct quantitative signature of this structure. If a lineage grows at net rate \(r\), raising the detection threshold from \(V_1\) to \(V_2\) shifts the detection time by
\begin{equation}
\Delta \tau_d
=
\frac{1}{r}\log\left(\frac{V_2}{V_1}\right).
\label{eq:threshold_shift}
\end{equation}
For \(r\approx0.69\,{\rm day}^{-1}\), the predicted delays for \(50\to400\) and \(400\to10{,}000\) copies/mL are approximately \(3.0\) and \(4.7\) days, respectively---capturing the logarithmic dependence and the correct order of magnitude, although the larger observed shift at high threshold may reflect slower effective growth, heterogeneity, or sampling effects. The reactivation process itself is intermittent: the once-every-\(5\)--\(8\)-day estimate of Pinkevych et al.~\cite{pinkevych2015latency} corresponds to an effective successful-reactivation rate \(\lambda\approx0.13\)--\(0.20\,{\rm day}^{-1}\). These observations motivate modeling \(\lambda(t)\) as a biological intensity shaped by reservoir size, reservoir activity, immune state, drug history, and intervention.

The central approximation of the paper separates the observed rebound time into a stochastic waiting time and a deterministic detection delay. If the first successful expanding lineage is established at time \(T_1\), begins with output \(v_0\), and grows with net rate \(r>0\), then detection occurs when
\begin{equation}
v_0\e^{r(t-T_1)}=\Vdet .
\end{equation}
The growth-to-detection delay is therefore
\begin{equation}
\td=
\frac{1}{r}\log\left(\frac{\Vdet}{v_0}\right).
\label{eq:td_section}
\end{equation}
Including a short eclipse phase \(\tecl\) before measurable shedding gives the total delay
\begin{equation}
\tdelay=\tecl+\td,
\label{eq:delay_total_section}
\end{equation}
and the single-founder decomposition
\begin{equation}
\Treb\approx T_1+\tdelay .
\label{eq:decomp_intro}
\end{equation}
Thus, \(T_1\) is controlled by the successful-reactivation hazard \(\lambda(t)\), whereas \(\tdelay\) is controlled by eclipse timing, net growth \(r\), founder output \(v_0\), and threshold \(\Vdet\).

This decomposition converts any mechanistic model for \(\lambda(t)\) into a rebound-time distribution. For a general time-dependent successful-reactivation rate,
\begin{equation}
S_{\rm reb}(t)
\approx
\begin{cases}
1, & t<\tw+\tdelay,\\[5pt]
\displaystyle
\exp\left[
-\int_{\tw}^{t-\tdelay}\lambda(s)\dd s
\right], & t\ge \tw+\tdelay .
\end{cases}
\label{eq:Sreb_shifted_general}
\end{equation}
Equivalently, the observed rebound clock is a delayed version of the hidden reactivation clock: rebound at time \(t\) is driven by successful reactivation near \(t-\tdelay\). If
\begin{equation}
H(y)=\int_{0}^{y}\lambda(\tw+s)\dd s ,
\end{equation}
then the mean waiting time to the first successful lineage is
\begin{equation}
\E[T_1-\tw]=\int_{0}^{\infty}\e^{-H(y)}\dd y ,
\label{eq:mean_waiting_general}
\end{equation}
and hence
\begin{equation}
\E[\Treb]
\approx
\tw+\tdelay+
\int_{0}^{\infty}
\exp\left[
-\int_{0}^{y}\lambda(\tw+s)\dd s
\right]\dd y .
\label{eq:mean_rebound_general_section}
\end{equation}
For constant \(\lambda(t)=\lambda\),
\begin{equation}
\E[\Treb]
\approx
\tw+\tecl+\frac{1}{\lambda}
+
\frac{1}{r}\log\left(\frac{\Vdet}{v_0}\right).
\label{eq:mean_constant_section}
\end{equation}
The same first-passage structure therefore applies across constant, pharmacokinetic, periodic, stochastic, and heterogeneous activation models.

\section{Model Formulation and Viral-Load Distribution}
\label{sec:model}

We now specify the stochastic viral-load model and derive its distributional structure. Successful events are modeled as an inhomogeneous Poisson process \(N(t)\), \(t\ge\tw\), with intensity \(\lambda(t)\ge0\). Throughout, \(\lambda(t)\) denotes the rate of successful established reactivation rather than mere transcriptional activation. If \(\lambda_{\rm act}(t)\) is the raw activation rate and \(p_{\rm est}(t)\) is the probability that an activated lineage establishes rather than undergoes early extinction, then
\begin{equation}
\lambda(t)=p_{\rm est}(t)\lambda_{\rm act}(t).
\label{eq:success_rate}
\end{equation}
Thus, immune clearance and early stochastic extinction enter through \(p_{\rm est}\) without changing the first-passage structure.

Each successful event at time \(T_i\) seeds a lineage with initial contribution \(v_0>0\) and net exponential growth rate \(r=g-c\), where \(g\) is intrinsic expansion and \(c\) is effective clearance or immune-mediated loss. For \(r>0\), the total viral load is the Poisson shot-noise process
\begin{equation}
V(t)
=
v_0\sum_{T_i\le t}\e^{r(t-T_i)}
=
\int_{\tw}^{t}v_0\e^{r(t-s)}N(\dd s).
\label{eq:V_sum}
\end{equation}
The rebound time is then the first passage of this process across \(\Vdet\), as defined in Eq.~\eqref{eq:Treb_firstpassage}.

A newly established lineage need not shed measurable plasma RNA immediately. We represent the eclipse phase by shifting the lineage contribution so that a founder established at \(s\) contributes \(v_0\e^{r(t-s-\tecl)}\) for \(t\ge s+\tecl\) and zero before. The total deterministic delay is therefore
\begin{equation}
\tdelay=\tecl+\frac{1}{r}\log\left(\frac{\Vdet}{v_0}\right).
\label{eq:tdelay_eclipse}
\end{equation}
In the technical derivations below we write formulas using \(\td\) alone; the eclipse is restored by the substitution
\begin{equation}
\td\longrightarrow \tdelay=\tecl+\td .
\label{eq:eclipse_substitution}
\end{equation}
This substitution shifts all means, quantiles, and survival curves later by \(\tecl\) without changing their shape.

\subsection{Physiological parameter values}
\label{sec:phys_params}

The numerical examples use the physiological parameter set in Table~\ref{tab:phys_params}. The successful-reactivation rate is taken from recurrent-reactivation estimates in HIV and SIV treatment-interruption studies~\cite{Pinkevych2019,Wu2020}. The early rebound growth rate is constrained by classical viral-dynamics measurements of virion clearance and infected-cell turnover~\cite{Perelson1996}. The founder output \(v_0\) is treated as an effective plasma-equivalent scale for one established lineage; because rebound time depends on \(v_0\) only logarithmically, order-of-magnitude uncertainty in \(v_0\) produces modest changes in predicted times~\cite{VanDorp2020}. The detection threshold is assay-dependent, while the washout parameters \(A_0\) and \(k_{\rm drug}\) represent residual drug suppression and pharmacokinetic decay.

\begin{table}[ht]
\centering
\caption{Physiological parameter values used in the rebound-time formulas. Ranges are drawn from HIV/SIV treatment-interruption and viral-dynamics studies~\cite{Perelson1996,Pinkevych2019,Wu2020,VanDorp2020,Hill2018}. The baseline column gives the representative value used in worked examples.}
\label{tab:phys_params}
\begin{tabularx}{\textwidth}{>{\raggedright\arraybackslash}p{0.30\textwidth}c c X}
\toprule
\textbf{Parameter} & \textbf{Symbol} & \textbf{Baseline} & \textbf{Physiological range / source} \\
\midrule
Successful reactivation rate & \(\lambda,\lambda_0\) & \(0.30\,\mathrm{day}^{-1}\) & \(0.17\)--\(0.54\) (HIV); \(0.5\)--\(2.1\) (SIV)~\cite{Pinkevych2019,Wu2020} \\
Net post-reactivation growth rate & \(r=g-c\) & \(0.69\,\mathrm{day}^{-1}\) & \(0.4\)--\(1.5\) (doubling time \(0.5\)--\(1.7\) d)~\cite{Perelson1996} \\
Per-founder initial output & \(v_0\) & \(1\,\mathrm{cp/mL}\) & \(0.1\)--\(10\) (effective establishment scale)~\cite{VanDorp2020} \\
Detection threshold & \(\Vdet\) & \(50\,\mathrm{cp/mL}\) & \(20\)--\(200\) (assay-dependent) \\
Growth delay & \(\td=r^{-1}\log(\Vdet/v_0)\) & \(5.7\,\mathrm{day}\) & derived \\
Eclipse phase & \(\tecl\) & \(1\,\mathrm{day}\) & \(0.5\)--\(2\) (RNA-appearance lag)~\cite{Perelson1996} \\
Initial ART suppression strength & \(A_0\) & \(1.0\) & \(0\)--\(1\) \\
Drug elimination rate & \(k_{\rm drug}\) & \(0.3\,\mathrm{day}^{-1}\) & \(0.1\)--\(0.7\)~\cite{Hill2018} \\
Time of ART interruption & \(\tw\) & \(0\,\mathrm{day}\) & reference origin \\
\bottomrule
\end{tabularx}
\end{table}

At baseline,
\begin{equation}
\td
=
\frac{1}{0.69}\log\left(\frac{50}{1}\right)
\approx 5.7\ {\rm days}.
\label{eq:tau_d_numeric}
\end{equation}
Together with \(\tecl\), this sets the deterministic floor below which detectable rebound is unlikely.

\subsection{Laplace transform and cumulants}

The distribution of \(V(t)\) is generally not elementary, but its Laplace transform is closed form because \(V(t)\) is a linear functional of a Poisson random measure.

\begin{proposition}[Laplace transform of the stochastic viral load]
For \(q\ge0\),
\begin{equation}
\Phi_t(q):=\E[\e^{-qV(t)}]
=
\exp\left\{
-\int_{\tw}^{t}\lambda(s)
\left[
1-\exp\left(-qv_0\e^{r(t-s)}\right)
\right]\dd s
\right\}.
\label{eq:laplace_general}
\end{equation}
\end{proposition}

\begin{proof}
The Laplace functional of a Poisson random measure gives, for nonnegative \(h\),
\[
\E\left[\exp\left(-\int h(s)N(\dd s)\right)\right]
=
\exp\left[-\int(1-\e^{-h(s)})\lambda(s)\dd s\right]
\]
\cite{Kingman1993,Feller1971}. Taking \(h(s)=qv_0\e^{r(t-s)}\ind_{\{\tw\le s\le t\}}\) yields Eq.~\eqref{eq:laplace_general}.
\end{proof}

The distribution function is obtained by Laplace inversion,
\begin{equation}
F_{V(t)}(x)
=
\Prob(V(t)\le x)
=
\Lcal^{-1}\!\left[\frac{\Phi_t(q)}{q}\right](x),
\label{eq:laplace_inversion}
\end{equation}
and the cumulants follow from \(\log\Phi_t(q)\):
\begin{equation}
\kappa_m(t)
=
v_0^m
\int_{\tw}^{t}
\lambda(s)\e^{mr(t-s)}\dd s,
\qquad m=1,2,\ldots .
\label{eq:cumulants_general}
\end{equation}
In particular,
\begin{equation}
\E[V(t)]
=
v_0\int_{\tw}^{t}\lambda(s)\e^{r(t-s)}\dd s,
\qquad
\Var[V(t)]
=
v_0^2\int_{\tw}^{t}\lambda(s)\e^{2r(t-s)}\dd s .
\label{eq:mean_var_general}
\end{equation}
For constant \(\lambda(t)=\lambda\), with \(\Delta=t-\tw\),
\begin{equation}
\E[V(t)]
=
\frac{v_0\lambda}{r}\left(\e^{r\Delta}-1\right),
\qquad
\Var[V(t)]
=
\frac{v_0^2\lambda}{2r}\left(\e^{2r\Delta}-1\right),
\label{eq:mean_var_constant}
\end{equation}
and
\begin{equation}
{\rm CV}^2(t)
=
\frac{\Var[V(t)]}{\E[V(t)]^2}
=
\frac{r}{2\lambda}
\frac{\e^{2r\Delta}-1}{(\e^{r\Delta}-1)^2}
\longrightarrow
\frac{r}{2\lambda}
\qquad
(\Delta\to\infty).
\label{eq:cv2}
\end{equation}
The nonzero late-time limit shows that, when successful reactivation is rare relative to viral expansion, individual trajectories remain highly dispersed and cannot be represented by a single deterministic mean trajectory.

For constant \(\lambda\), the shot-noise structure can also be written explicitly as a compound-Poisson variable. Conditional on \(N(t)=n\), the event times are independent and uniform on \([\tw,t]\). Thus each contribution has the form \(Y=v_0\e^{r(t-U)}\), with \(U\sim{\rm Uniform}(\tw,t)\), so \(Y\in[v_0,v_0\e^{r\Delta}]\) and \(f_Y(y)=1/(r\Delta y)\). Hence
\begin{equation}
V(t)\stackrel{d}{=}\sum_{i=1}^{N_\Delta}Y_i,
\qquad
N_\Delta\sim{\rm Poisson}(\lambda\Delta),
\label{eq:compound_poisson}
\end{equation}
where the \(Y_i\) are independent copies of \(Y\). This representation is used in the first-passage analysis of Appendix~\ref{app:first_passage_bounds}; representative realizations are shown in Figure~\ref{fig:schematic}.

Figure~\ref{fig:schematic} shows five independent realizations of the shot-noise viral load of Eq.~\eqref{eq:V_sum}, simulated at the baseline parameters of Table~\ref{tab:phys_params} (\(\lambda=0.30\,{\rm day}^{-1}\), \(r=0.69\,{\rm day}^{-1}\), \(v_0=1\), \(V_{\rm det}=50\)). Reactivation times are drawn from a homogeneous Poisson process and each founder lineage is propagated exactly, so that the plotted curves are sample paths of the exact process rather than of the single-founder approximation. The filled circle on each curve marks its first-passage time \(T_{\rm reb}=\inf\{t:V(t)\ge V_{\rm det}\}\). Although the trajectories share identical parameters, they cross the threshold across a wide spread of times; this is the visual counterpart of the persistent late-time dispersion \({\rm CV}^2\to r/2\lambda\) of Eq.~\eqref{eq:cv2}, and it shows directly why rebound timing cannot be reduced to a single deterministic mean trajectory.

\begin{figure}[tbp]
\centering
\includegraphics[width=0.78\linewidth]{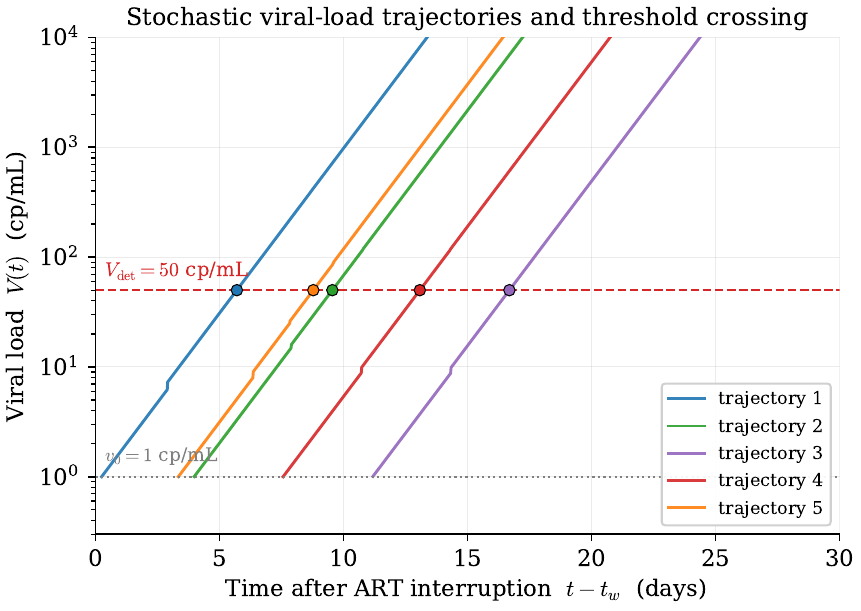}
\caption{Representative shot-noise viral-load trajectories \(V(t)=v_0\sum_{T_i\le t}e^{r(t-T_i)}\) from Eq.~\eqref{eq:V_sum}, for \(\lambda=0.30\,{\rm day}^{-1}\), \(v_0=1\), \(r=0.69\,{\rm day}^{-1}\), and \(V_{\rm det}=50\) copies\,mL\(^{-1}\). Filled circles mark first-passage rebound times \(T_{\rm reb}=\inf\{t\ge t_w:V(t)\ge V_{\rm det}\}\). The wide spread of crossing times at fixed parameters reflects the persistent dispersion of Eq.~\eqref{eq:cv2}.}
\label{fig:schematic}
\end{figure}

\section{From Viral-Load Distribution to Rebound-Time Law}
\label{sec:exact_fp}

The viral-load distribution determines the rebound-time law. When \(r>0\), each established lineage grows monotonically, so \(V(t)\) increases between reactivation events and jumps upward at event times. Therefore,
\begin{equation}
\{\Treb\le t\}=\{V(t)\ge\Vdet\},
\end{equation}
and the exact monotone first-passage relation is
\begin{equation}
\Prob(\Treb>t)=F_{V(t)}(\Vdet^-),
\qquad
\Prob(\Treb\le t)=1-F_{V(t)}(\Vdet^-).
\label{eq:T_survival_exact}
\end{equation}
Together with Eqs.~\eqref{eq:laplace_general}--\eqref{eq:laplace_inversion}, this gives an exact semi-analytical rebound-time distribution, up to numerical Laplace inversion. If net growth becomes negative or strongly time dependent, \(V(t)\) need not be monotone and the equality above fails; then a genuinely path-dependent first-passage calculation is required. The monotone regime \(r>0\) is the natural baseline for early post-interruption rebound.

\subsection{Bias of mean-field threshold crossing}

A common deterministic surrogate defines \(\bar T\) by
\begin{equation}
\E[V(\bar T)]=\Vdet .
\label{eq:meanfield_def}
\end{equation}
For constant \(\lambda\), Eq.~\eqref{eq:mean_var_constant} gives
\begin{equation}
\bar T
=
\tw+
\frac{1}{r}
\log\left(1+\frac{r\Vdet}{v_0\lambda}\right).
\label{eq:meanfield_constant}
\end{equation}
This time summarizes the ensemble-average viral burden but is generally not the mean or median first-passage time. The reason is structural: \(\E[V(t)]\) averages individuals who have not yet reactivated together with a minority of early rebound trajectories carrying large viral loads. In the rare-reactivation regime, this tail pulls the mean-load crossing earlier than the typical individual crossing. As \(\lambda\to0\),
\begin{equation}
\bar T-\tw
\sim
\frac{1}{r}
\log\left(\frac{r\Vdet}{v_0\lambda}\right),
\end{equation}
whereas the stochastic waiting time to the first successful lineage scales as \(1/\lambda\). Mean-field threshold crossing therefore systematically underestimates typical rebound times when reactivation is rare, as illustrated in Figure~\ref{fig:mean_vs_fp}.

Figure~\ref{fig:mean_vs_fp} contrasts the mean-field threshold crossing of Eq.~\eqref{eq:meanfield_constant} with the true first-passage law in a rare-reactivation regime (\(\lambda=0.10\,{\rm day}^{-1}\)). The solid curve is a \(2\times10^{4}\)-trajectory simulation of the exact rebound survival; the dashed curve is the single-founder law of Eq.~\eqref{eq:sf_survival_general}, which the simulation tracks almost exactly. The mean-field crossing \(\bar T\approx8.5\) days from Eq.~\eqref{eq:meanfield_constant} precedes the simulated median of \(\approx12.6\) days by roughly four days, making concrete the systematic early bias of the ensemble-mean surrogate: \(\mathbb E[V(t)]\) is inflated by a minority of early high-growth lineages while most individuals remain below detection.

\begin{figure}[tbp]
\centering
\includegraphics[width=0.78\linewidth]{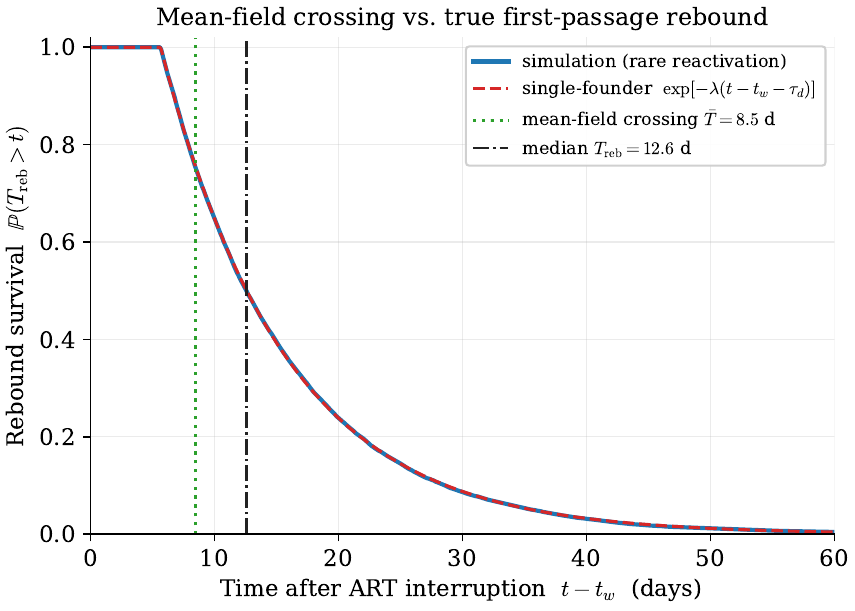}
\caption{Mean-field threshold crossing versus stochastic first passage in a rare-reactivation regime. The mean-field crossing time defined by \(\mathbb E[V(\bar T)]=V_{\rm det}\) (Eq.~\eqref{eq:meanfield_constant}) occurs earlier than the typical rebound time because the ensemble mean is dominated by rare early high-growth trajectories. The single-founder survival of Eq.~\eqref{eq:sf_survival_general} (dashed) matches the exact simulation (solid).}
\label{fig:mean_vs_fp}
\end{figure}

\subsection{Single-founder approximation}
\label{sec:single_founder}

The exact first-passage relation in Eq.~\eqref{eq:T_survival_exact} is mathematically complete, but it is not always the most transparent form for interpretation or inference. The central simplification used in the closed-form results below is the single-founder approximation. In this approximation, detectable rebound is assumed to be dominated by the earliest successful expanding lineage. Contributions from later lineages that are individually below threshold are retained only as a cooperative correction.

Let \(T_1\) denote the first successful established reactivation time after washout. If a founder lineage is established at time \(s\), begins with effective plasma contribution \(v_0\), and grows with net rate \(r>0\), then in the absence of an eclipse phase it becomes detectable when
\begin{equation}
v_0 \e^{r(t-s)}=\Vdet .
\label{eq:founder_detection_condition}
\end{equation}
Solving Eq.~\eqref{eq:founder_detection_condition} gives the growth-to-detection delay
\begin{equation}
\td
=
\frac{1}{r}
\log\left(\frac{\Vdet}{v_0}\right).
\label{eq:td_single_founder}
\end{equation}
If an eclipse phase of duration \(\tau_e\) precedes measurable plasma production, the total delay from successful establishment to detection is
\begin{equation}
\tdelay
=
\tau_e+\td
=
\tau_e+
\frac{1}{r}
\log\left(\frac{\Vdet}{v_0}\right).
\label{eq:tdelay_single_founder}
\end{equation}
Thus, the observed rebound time separates into a hidden stochastic waiting time and a deterministic detection delay,
\begin{equation}
\Treb\approx T_1+\tdelay .
\label{eq:single_founder_T}
\end{equation}
This decomposition is the mathematical basis for the shifted-hazard laws derived below. It states that the clinical rebound clock is a delayed version of the hidden successful-reactivation clock.

\begin{proposition}[Single-founder rebound survival]
Under the single-founder approximation, the rebound-survival probability is
\begin{equation}
\Prob(\Treb>t)
\approx
\begin{cases}
1, & t<\tw+\tdelay,\\[6pt]
\displaystyle
\exp\left[-\int_{\tw}^{t-\tdelay}\lambda(s)\dd s\right],
& t\ge \tw+\tdelay,
\end{cases}
\label{eq:sf_survival_general}
\end{equation}
with density
\begin{equation}
f_{\rm reb}(t)
\approx
\lambda(t-\tdelay)
\exp\left[-\int_{\tw}^{t-\tdelay}\lambda(s)\dd s\right]
\ind_{\{t\ge \tw+\tdelay\}} .
\label{eq:sf_density_general}
\end{equation}
\end{proposition}

\begin{proof}
A successful founder established at time \(s\) becomes detectable approximately at \(s+\tdelay\). Therefore, rebound by time \(t\) requires at least one successful founder before \(t-\tdelay\). Equivalently, under the single-founder approximation, the event \(\{\Treb>t\}\) is approximated by the event that no successful founder has occurred in the interval \([\tw,t-\tdelay]\). Since successful founder events form an inhomogeneous Poisson process with intensity \(\lambda(s)\),
\begin{equation}
\Prob(\Treb>t)
\approx
\Prob(T_1>t-\tdelay)
=
\exp\left[-\int_{\tw}^{t-\tdelay}\lambda(s)\dd s\right],
\qquad t\ge \tw+\tdelay .
\end{equation}
For \(t<\tw+\tdelay\), even a founder established immediately at \(\tw\) has not had sufficient time to reach the detection threshold, so \(\Prob(\Treb>t)\approx1\). Differentiating the survival function gives Eq.~\eqref{eq:sf_density_general}.
\end{proof}

Equation~\eqref{eq:sf_survival_general} has a direct biological interpretation. The probability that a participant remains below detection at time \(t\) is approximately the probability that no successful reactivation event occurred early enough to grow to \(\Vdet\) by time \(t\). Thus, the observed rebound hazard at time \(t\) is controlled by the hidden successful-reactivation intensity at the earlier time \(t-\tdelay\).

For a constant successful-reactivation rate, \(\lambda(t)=\lambda\), Eq.~\eqref{eq:sf_survival_general} reduces to the shifted exponential survival law
\begin{equation}
\Prob(\Treb>t)
\approx
\begin{cases}
1, & t<\tw+\tdelay,\\[6pt]
\exp[-\lambda(t-\tw-\tdelay)],
& t\ge \tw+\tdelay .
\end{cases}
\label{eq:sf_constant_survival}
\end{equation}
Equivalently,
\begin{equation}
\Treb
\approx
\tw+\tdelay+E,
\qquad
E\sim {\rm Exp}(\lambda).
\label{eq:sf_constant_decomposition}
\end{equation}
Taking expectations gives
\begin{equation}
\E[\Treb]
\approx
\tw+\tdelay+\frac{1}{\lambda}.
\label{eq:sf_constant_mean_delay}
\end{equation}
Using Eq.~\eqref{eq:tdelay_single_founder}, this becomes
\begin{equation}
\E[\Treb]
\approx
\tw+\tau_e+\frac{1}{\lambda}
+
\frac{1}{r}
\log\left(\frac{\Vdet}{v_0}\right).
\label{eq:sf_constant_mean_full}
\end{equation}
Hence, the familiar constant-rate mean rebound time is not an independent assumption. It follows directly from the single-founder shifted-hazard approximation. The term \(1/\lambda\) is the mean waiting time for the first successful established lineage, while \(\tau_e+r^{-1}\log(\Vdet/v_0)\) is the deterministic delay from establishment to detectability.

The approximation also has a definite direction. If any founder occurs before \(t-\tdelay\), that lineage alone reaches \(\Vdet\) by time \(t\). Therefore
\begin{equation}
\{\Treb>t\}\subseteq\{N(t-\tdelay)-N(\tw)=0\},
\qquad
\Prob(\Treb>t)\le
\exp\left[-\int_{\tw}^{t-\tdelay}\lambda(s)\dd s\right].
\label{eq:sf_upper_bound}
\end{equation}
Thus, the single-founder survival is an upper bound on the true survival probability. It becomes exact when threshold crossing by several individually sub-detectable lineages is negligible.

A useful measure of the possible cooperative correction is
\begin{equation}
\epsilon(t)
=
\int_{t-\tdelay}^{t}\lambda(s)\dd s ,
\label{eq:epsilon_coop}
\end{equation}
the expected number of successful founders within one detection-delay window. When \(\epsilon(t)\ll1\), few founders arise during the time required for one lineage to become detectable, and the earliest successful founder typically dominates. When \(\epsilon(t)\gtrsim1\), several recent lineages may coexist below threshold and their summed viral output can advance rebound relative to \(T_1+\tdelay\).

For a constant rate, \(\epsilon(t)=\lambda\tdelay\). This dimensionless quantity compares the growth-to-detection timescale with the mean waiting time for successful reactivation. Small \(\lambda\tdelay\) corresponds to the rare-reactivation regime, in which the single-founder approximation is most accurate. Larger \(\lambda\tdelay\) increases the opportunity for cooperative crossing, but exponential growth still gives the earliest founder a multiplicative advantage over later founders. Consequently, cooperative effects generally shift rebound earlier rather than changing the qualitative structure of the shifted-hazard law.
At the baseline values \(\lambda=0.30\,{\rm day}^{-1}\) and \(\td=5.7\) days, \(\lambda\td\approx1.7\), so the conservative condition \(\epsilon\ll1\) is not strictly satisfied. We therefore compare the approximation with full stochastic simulation of the exact process, allowing any number of sub-detectable lineages to cross cooperatively. Each simulation draws Poisson reactivation times, assigns to each event a lineage \(v_0\e^{r(t-T_i)}\), and records the first realization-by-realization crossing of \(\Vdet\). No single-founder assumption is imposed in the simulation, so the comparison isolates the cooperative-crossing error.

\begin{table}[ht]
\centering
\caption{Accuracy of the single-founder (SF) median rebound time against full cooperative simulation (\(6\times10^4\) trajectories per row, fixed seed; \(r=0.69\,\mathrm{day}^{-1}\), \(v_0=1\), \(\Vdet=50\)).}
\label{tab:sf_accuracy}
\begin{tabular}{ccccc}
\toprule
\(\lambda\) (day\(^{-1}\)) & \(\lambda\td\) & SF median (d) & Simulated median (d) & Relative error \\
\midrule
\(0.17\) & \(0.96\) & \(9.8\) & \(9.5\) & \(-2.9\%\) \\
\(0.20\) & \(1.13\) & \(9.1\) & \(8.8\) & \(-3.7\%\) \\
\(0.30\) & \(1.70\) & \(8.0\) & \(7.6\) & \(-5.1\%\) \\
\(0.54\) & \(3.06\) & \(7.0\) & \(6.3\) & \(-9.8\%\) \\
\bottomrule
\end{tabular}
\end{table}

The single-founder formula overestimates the median by about \(5\%\) at \(\lambda\td\approx1.7\) and by about \(10\%\) at \(\lambda\td\approx3\). This robustness follows from exponential growth: the earliest founder maintains a multiplicative advantage \(\e^{r\Delta T}\) over a later founder born \(\Delta T\) afterward, so cooperative crossing usually advances detection only modestly. We therefore use the single-founder formulas as the principal closed-form results, noting that they slightly overestimate rebound times at high \(\lambda\); the leading cooperative correction is derived in Appendix~\ref{app:two_founder}.
\section{Explicit Activation Profiles and Rebound-Time Laws}
\label{sec:activation_profiles}

The previous sections showed that, in the single-founder approximation, any successful-reactivation intensity \(\lambda(t)\) generates a rebound-time law through the shifted cumulative hazard
\begin{equation}
\Prob(T_{\rm reb}>t)
\approx
\exp\!\left[-\int_{\tw}^{t-\tau_{\rm det}}\lambda(s)\dd s\right],
\qquad
t\ge \tw+\tau_{\rm det},
\label{eq:general_shifted_survival_profiles}
\end{equation}
with \(\Prob(T_{\rm reb}>t)\approx1\) for \(t<\tw+\tau_{\rm det}\). Here
\begin{equation}
\tau_{\rm det}
=
\tau_e+
\frac{1}{r}
\log\!\left(\frac{V_{\rm det}}{v_0}\right)
\label{eq:tau_det_profiles}
\end{equation}
is the deterministic delay from successful activation to clinical detectability. In the present paper, \(\lambda(t)\) denotes the intensity of \emph{successful established} reactivation. If \(\lambda_{\rm act}(t)\) is the raw activation intensity used in earlier activation-survival models and \(p_{\rm est}(t)\) is the probability that an activated lineage establishes rather than undergoes early extinction, then
\begin{equation}
\lambda(t)=p_{\rm est}(t)\lambda_{\rm act}(t).
\label{eq:lambda_successful_adjusted}
\end{equation}
Thus, the Poisson intensities used previously for constant, periodic, stochastic, and washout-driven reactivation remain valid after this biological reinterpretation; when \(p_{\rm est}\) is constant it is absorbed into the effective baseline rate \(\lambda_0\)~\cite{Taye2025CM}.

\subsection{Constant successful reactivation}
\label{sec:constant}

For the baseline case \(\lambda(t)=\lambda\), the first successful established lineage satisfies
\begin{equation}
T_1-\tw\sim {\rm Exp}(\lambda),
\end{equation}
and therefore
\begin{equation}
T_{\rm reb}\approx \tw+\tau_{\rm det}+E,
\qquad E\sim{\rm Exp}(\lambda).
\label{eq:constant_shifted_exp_compact}
\end{equation}
The rebound survival function, density, mean, variance, and \(p\)-th quantile are
\begin{align}
\Prob(T_{\rm reb}>t)
&=
\begin{cases}
1, & t<\tw+\tau_{\rm det},\\
\exp[-\lambda(t-\tw-\tau_{\rm det})],
& t\ge\tw+\tau_{\rm det},
\end{cases}
\label{eq:constant_survival_compact}
\\[4pt]
f_{\rm reb}(t)
&=
\lambda e^{-\lambda(t-\tw-\tau_{\rm det})}
\mathbf{1}_{\{t\ge\tw+\tau_{\rm det}\}},
\label{eq:constant_density_compact}
\\[4pt]
\E[T_{\rm reb}]
&\approx
\tw+\tau_{\rm det}+\frac{1}{\lambda},
\qquad
\Var(T_{\rm reb})\approx\frac{1}{\lambda^2},
\label{eq:constant_mean_var_compact}
\\[4pt]
Q_p
&=
\tw+\tau_{\rm det}
-\frac{1}{\lambda}\log(1-p).
\label{eq:constant_quantile_compact}
\end{align}
Thus, constant successful reactivation produces a shifted exponential rebound-time distribution: \(\tau_{\rm det}\) translates the distribution, while \(1/\lambda\) controls its width.

The earlier paper estimated raw reactivation rates in the range \(\lambda_{\rm act}\approx0.17\)--\(0.54\,{\rm day}^{-1}\) for HIV and \(0.5\)--\(2.1\,{\rm day}^{-1}\) for SIV~\cite{Pinkevych2019,Wu2020,Taye2025CM}. In the present first-passage setting these values should be interpreted as upper-scale activation intensities; the effective successful rate is \(\lambda=p_{\rm est}\lambda_{\rm act}\). For example, with \(V_{\rm det}=50\), \(v_0=1\), and \(r=0.69\,{\rm day}^{-1}\), one obtains \(\tau_d\simeq5.7\) days. Taking \(\lambda=0.30\,{\rm day}^{-1}\) and \(\tau_e=0\) gives
\begin{equation}
\E[T_{\rm reb}]\approx5.7+\frac{1}{0.30}\approx9.0\ {\rm days},
\qquad
Q_{0.5}\approx5.7+\frac{\log2}{0.30}\approx8.0\ {\rm days}.
\end{equation}
This baseline is intentionally fast and should not be claimed to reproduce the \(16\)-day median reported for the \(50\)-copy threshold~\cite{gunst2025ati}. A data-calibrated interpretation instead uses a slower effective growth rate, for example \(r_{\rm eff}\approx0.33\,{\rm day}^{-1}\), together with \(\lambda_{\rm eff}\approx0.17\)--\(0.20\,{\rm day}^{-1}\) and a nonzero eclipse phase, yielding rebound medians in the observed ATI range.

\subsection{ART-washout-dependent reactivation}
\label{sec:washout}

The washout model in the earlier paper used the pharmacokinetic form \(\lambda_{\rm act}(t)=\lambda_0(1-e^{-k_{\rm drug}t})\). We use the shifted and generalized successful-reactivation form
\begin{equation}
\lambda(t)
=
\lambda_0\left[1-A_0e^{-k_{\rm drug}(t-\tw)}\right],
\qquad t\ge\tw,
\label{eq:washout_lambda_compact}
\end{equation}
where \(\lambda_0\) is the unsuppressed successful-reactivation rate, \(A_0\in[0,1]\) is the initial suppression strength, and \(k_{\rm drug}\) is the elimination rate. Defining
\begin{equation}
y=t-\tw-\tau_{\rm det},
\end{equation}
the shifted cumulative hazard is
\begin{equation}
\Lambda_{\rm w}(y)
=
\lambda_0\left[
y-\frac{A_0}{k_{\rm drug}}\left(1-e^{-k_{\rm drug}y}\right)
\right],
\qquad y\ge0.
\label{eq:washout_cumhaz_compact}
\end{equation}
Hence
\begin{align}
\Prob(T_{\rm reb}>t)
&\approx
\exp[-\Lambda_{\rm w}(y)],
\qquad y\ge0,
\label{eq:washout_survival_compact}
\\[4pt]
f_{\rm reb}(t)
&\approx
\lambda_0\left[1-A_0e^{-k_{\rm drug}y}\right]
\exp[-\Lambda_{\rm w}(y)]
\mathbf{1}_{\{y\ge0\}}.
\label{eq:washout_density_compact}
\end{align}
This expression reduces to the constant-rate law when \(A_0=0\) or \(k_{\rm drug}\to\infty\). When \(A_0=1\) and \(k_{\rm drug}y\ll1\),
\begin{equation}
\Lambda_{\rm w}(y)\approx \frac{1}{2}\lambda_0k_{\rm drug}y^2,
\end{equation}
so the early post-delay waiting time is approximately Rayleigh:
\begin{equation}
\Prob(T_{\rm reb}>t)
\approx
\exp\!\left[-\frac{1}{2}\lambda_0k_{\rm drug}y^2\right],
\qquad
\E[T_{\rm reb}]
\approx
\tw+\tau_{\rm det}
+\sqrt{\frac{\pi}{2\lambda_0k_{\rm drug}}}.
\label{eq:washout_rayleigh_compact}
\end{equation}
Thus, slow washout delays rebound on the scale \((\lambda_0k_{\rm drug})^{-1/2}\), rather than the constant-rate scale \(\lambda_0^{-1}\).

The exact mean is
\begin{equation}
\E[T_{\rm reb}]
=
\tw+\tau_{\rm det}
+
\int_0^\infty
\exp[-\Lambda_{\rm w}(y)]\,\dd y,
\label{eq:washout_mean_integral_compact}
\end{equation}
with the convergent series
\begin{equation}
\int_0^\infty e^{-\Lambda_{\rm w}(y)}\dd y
=
\frac{e^{c}}{\lambda_0}
\sum_{n=0}^{\infty}
\frac{(-c)^n}{n!}
\frac{1}{1+n k_{\rm drug}/\lambda_0},
\qquad
c=\frac{\lambda_0A_0}{k_{\rm drug}} .
\label{eq:washout_mean_series_compact}
\end{equation}
For \(\lambda_0=0.30\,{\rm day}^{-1}\), \(A_0=1\), \(k_{\rm drug}=0.30\,{\rm day}^{-1}\), and \(\tau_{\rm det}=5.7\) days, the waiting integral is approximately \(5.7\) days, giving
\begin{equation}
\E[T_{\rm reb}]\approx 11.4\ {\rm days}.
\end{equation}
Halving the clearance rate to \(k_{\rm drug}=0.15\,{\rm day}^{-1}\) increases the waiting term to about \(7.3\) days and gives \(\E[T_{\rm reb}]\approx13.0\) days. The dependence of this survival law on \(k_{\rm drug}\) is shown in Figure~\ref{fig:washout}.

Figure~\ref{fig:washout} plots the washout survival law of Eq.~\eqref{eq:washout_survival_compact}, with shifted cumulative hazard \(\Lambda_w(y)\) from Eq.~\eqref{eq:washout_cumhaz_compact}, for three elimination rates at \(A_0=1\), against the constant-rate limit of Eq.~\eqref{eq:constant_survival_compact}. All curves share the deterministic floor \(t_w+\tau_d\), below which rebound is suppressed. Slower drug clearance (smaller \(k_{\rm drug}\)) holds the early cumulative hazard down and shifts the survival curve later; as \(k_{\rm drug}\to\infty\) residual drug decays instantly and the curve collapses onto the constant-rate shifted-exponential. The early-time Rayleigh scaling \((\lambda_0 k_{\rm drug})^{-1/2}\) of Eq.~\eqref{eq:washout_rayleigh_compact} is visible as the flattened shoulder just past the floor for the slowest-washout (\(k_{\rm drug}=0.1\)) curve.

\begin{figure}[tbp]
\centering
\includegraphics[width=0.78\linewidth]{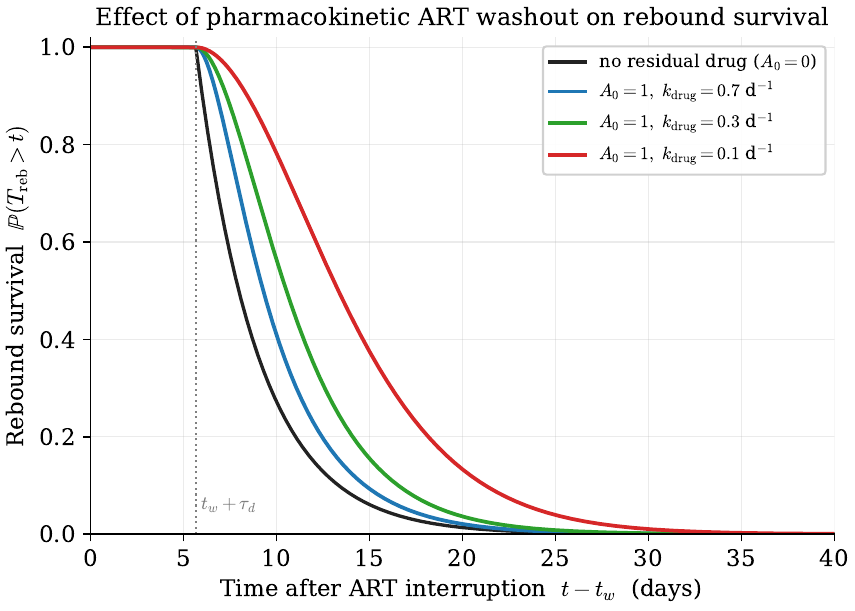}
\caption{Effect of pharmacokinetic ART washout on rebound survival, from the shifted cumulative hazard \(\Lambda_w(y)\) of Eq.~\eqref{eq:washout_cumhaz_compact}. Slower drug elimination suppresses the early cumulative hazard and shifts rebound later; rapid washout approaches the constant-rate shifted-exponential limit of Eq.~\eqref{eq:constant_survival_compact}.}
\label{fig:washout}
\end{figure}

The corresponding ensemble-mean viral load is
\begin{equation}
\E[V(t)]
=
v_0\lambda_0
\left[
\frac{e^{r\Delta}-1}{r}
-
A_0e^{r\Delta}
\frac{1-e^{-(r+k_{\rm drug})\Delta}}{r+k_{\rm drug}}
\right],
\qquad
\Delta=t-\tw.
\label{eq:mean_washout_correct_compact}
\end{equation}
This mean-field expression is useful for comparison but should not be confused with the first-passage law in Eq.~\eqref{eq:washout_survival_compact}.

\subsection{Periodically modulated immune activation}
\label{sec:periodic}

The earlier sinusoidal model used \(\lambda_{\rm act}(t)=\lambda_0+A\sin(\omega t)\). To preserve nonnegativity and separate the mean rate from the relative modulation amplitude, we write
\begin{equation}
\lambda(t)
=
\lambda_0
\left[
1+a\sin(\omega(t-\tw)+\phi)
\right],
\qquad |a|<1,
\label{eq:periodic_lambda_compact}
\end{equation}
where \(a=A/\lambda_0\) is dimensionless and \(\phi\) is the immune phase at interruption. With \(y=t-\tw-\tau_{\rm det}\),
\begin{equation}
\Lambda_{\rm p}(y)
=
\lambda_0y
+
\frac{\lambda_0a}{\omega}
\left[
\cos\phi-\cos(\omega y+\phi)
\right],
\label{eq:periodic_cumhaz_compact}
\end{equation}
and therefore
\begin{align}
\Prob(T_{\rm reb}>t)
&\approx
\exp[-\Lambda_{\rm p}(y)],
\qquad y\ge0,
\label{eq:periodic_survival_compact}
\\[4pt]
f_{\rm reb}(t)
&\approx
\lambda_0[1+a\sin(\omega y+\phi)]
\exp[-\Lambda_{\rm p}(y)]
\mathbf{1}_{\{y\ge0\}}.
\label{eq:periodic_density_compact}
\end{align}
The detection delay shifts immune-phase information forward in time: an activation peak at the hidden reactivation scale produces increased detectable rebound approximately \(\tau_{\rm det}\) later.

For \(|a|\ll1\), expansion of the survival integral gives
\begin{equation}
\E[T_{\rm reb}]
\approx
\tw+\tau_{\rm det}
+\frac{1}{\lambda_0}
-\frac{a}{\omega}\cos\phi
+
\frac{a\lambda_0}{\omega}
\frac{\lambda_0\cos\phi-\omega\sin\phi}
{\lambda_0^2+\omega^2}
+O(a^2).
\label{eq:periodic_mean_compact}
\end{equation}
The phase correction is largest when the immune period is comparable to the intrinsic waiting time \(1/\lambda_0\); rapid oscillations average out. For example, with \(\lambda_0=0.30\,{\rm day}^{-1}\), \(a=0.5\), \(\omega=2\pi/7\,{\rm day}^{-1}\), \(\phi=0\), and \(\tau_{\rm det}=5.7\) days,
\begin{equation}
\E[T_{\rm reb}]
\approx
5.7+3.33-0.557+0.057
\approx
8.5\ {\rm days}.
\end{equation}
This value illustrates phase-dependent acceleration, not a universal prediction: changing \(\phi\) can delay rather than advance rebound.

\subsection{Random reactivation rates and Cox-process rebound}
\label{sec:cox}

The previous stochastic-rate model used the additive form \(\lambda_{\rm act}(t)=\lambda_0+\eta(t)\), which is analytically transparent but may become negative for large fluctuations. For a first-passage theory, it is preferable to use a nonnegative Cox intensity,
\begin{equation}
\lambda(t)
=
\lambda_0
\exp\!\left[
\sigma X(t)-\frac{1}{2}\sigma^2
\right],
\label{eq:lognormal_intensity_compact}
\end{equation}
where \(X(t)\) is a mean-zero stochastic process with unit marginal variance. This preserves \(\E[\lambda(t)]=\lambda_0\) while allowing immune bursts and suppressed intervals.

Conditioned on a realization of \(\lambda(t)\),
\begin{equation}
\Prob(T_{\rm reb}>t\mid\lambda)
\approx
\exp\!\left[
-\int_{\tw}^{t-\tau_{\rm det}}\lambda(s)\dd s
\right].
\label{eq:cox_conditional_compact}
\end{equation}
With
\begin{equation}
\Xi_y=\int_{\tw}^{\tw+y}\lambda(s)\dd s,
\qquad y=t-\tw-\tau_{\rm det},
\end{equation}
the unconditional survival function is the Laplace transform of the integrated random hazard:
\begin{equation}
\Prob(T_{\rm reb}>t)
\approx
\E[e^{-\Xi_y}].
\label{eq:cox_survival_compact}
\end{equation}
If \(\Xi_y\) has mean \(\lambda_0y\) and variance \(\sigma_\Xi^2(y)\), the second-cumulant approximation gives
\begin{equation}
\Prob(T_{\rm reb}>t)
\approx
\exp\!\left[-\lambda_0y+\frac{1}{2}\sigma_\Xi^2(y)\right].
\label{eq:cox_cumulant_compact}
\end{equation}
For short-correlated fluctuations, \(\sigma_\Xi^2(y)\approx Dy\), and hence
\begin{equation}
\E[T_{\rm reb}]
\approx
\tw+\tau_{\rm det}
+
\frac{1}{\lambda_0-D/2},
\qquad D<2\lambda_0 .
\label{eq:cox_mean_compact}
\end{equation}
This is the first-passage analogue of the earlier waiting-time correction
\[
\E[1/\lambda(t)]
\approx
\frac{1}{\lambda_0}
\left(1+\frac{\sigma_\lambda^2}{\lambda_0^2}\right),
\]
but expressed through the integrated hazard rather than through instantaneous inter-event times. For \(\lambda_0=0.30\,{\rm day}^{-1}\), \(D=0.10\,{\rm day}^{-1}\), and \(\tau_{\rm det}=5.7\) days,
\begin{equation}
\E[T_{\rm reb}]
\approx
5.7+\frac{1}{0.30-0.05}
=
9.7\ {\rm days}.
\end{equation}
Thus, hazard variability increases the expected rebound time at the population level, because low-hazard intervals contribute disproportionately to long survival.

\subsection{Heterogeneous reservoir classes}
\label{sec:heterogeneous}

Reservoir heterogeneity can be incorporated by assigning each class \(j=1,\ldots,J\) its own successful-reactivation intensity, founder output, and growth rate. If class \(j\) has \(\lambda_j(t)\), \(v_j\), and \(r_j>0\), then
\begin{equation}
V(t)
=
\sum_{j=1}^{J}
\sum_{T_{ij}\le t}
v_j e^{r_j(t-T_{ij})}.
\label{eq:hetero_V_compact}
\end{equation}
The exact Laplace transform is
\begin{equation}
\Phi_t(q)
=
\exp\!\left\{
-\sum_{j=1}^{J}
\int_{\tw}^{t}
\lambda_j(s)
\left[
1-e^{-qv_j e^{r_j(t-s)}}
\right]\dd s
\right\}.
\label{eq:hetero_laplace_compact}
\end{equation}
The first-passage approximation is obtained by assigning each class its own detection delay
\begin{equation}
\tau_{{\rm det},j}
=
\tau_{e,j}
+
\frac{1}{r_j}
\log\!\left(\frac{V_{\rm det}}{v_j}\right),
\label{eq:hetero_delay_compact}
\end{equation}
so that
\begin{equation}
\Prob(T_{\rm reb}>t)
\approx
\exp\!\left[
-\sum_{j=1}^{J}
\int_{\tw}^{t-\tau_{{\rm det},j}}
\lambda_j(s)\dd s
\right],
\label{eq:hetero_survival_compact}
\end{equation}
where an integral is set to zero if \(t-\tau_{{\rm det},j}<\tw\). The density is
\begin{equation}
f_{\rm reb}(t)
\approx
\left[
\sum_{j=1}^{J}
\lambda_j(t-\tau_{{\rm det},j})
\mathbf{1}_{\{t\ge\tw+\tau_{{\rm det},j}\}}
\right]
\Prob(T_{\rm reb}>t).
\label{eq:hetero_density_compact}
\end{equation}
For constant class-specific rates,
\begin{equation}
\Prob(T_{\rm reb}>t)
\approx
\exp\!\left[
-\sum_{j=1}^{J}
\lambda_j
(t-\tw-\tau_{{\rm det},j})_+
\right].
\label{eq:hetero_constant_survival_compact}
\end{equation}
Thus, early rebound need not be seeded by the largest reservoir class. A rare class can dominate if it has a sufficiently high establishment rate, large founder output, or rapid post-reactivation growth. The probability that class \(j\) seeds the first detectable rebound event is
\begin{equation}
\Prob(J^*=j)
=
\int_{\tw}^{\infty}
\lambda_j(t-\tau_{{\rm det},j})
\mathbf{1}_{\{t\ge\tw+\tau_{{\rm det},j}\}}
\exp\!\left[
-\sum_{\ell=1}^{J}
\int_{\tw}^{t-\tau_{{\rm det},\ell}}
\lambda_\ell(s)\dd s
\right]\dd t.
\label{eq:trigger_prob_compact}
\end{equation}
Figure~\ref{fig:heterogeneous} illustrates this competing-risks structure for reservoir classes with different activation and growth parameters.

In detail, the figure uses two classes: an abundant slow grower (class~1: \(\lambda_1=0.25\,{\rm day}^{-1}\), \(r_1=0.5\,{\rm day}^{-1}\), \(\tau_{d,1}\approx7.8\) days) and a rare fast, high-output grower (class~2: \(\lambda_2=0.05\,{\rm day}^{-1}\), \(r_2=1.0\,{\rm day}^{-1}\), \(v_2=10\), \(\tau_{d,2}\approx1.6\) days), combined through the constant-rate survival of Eq.~\eqref{eq:hetero_constant_survival_compact}. The dashed curves are the single-class survivals and the solid curve is their product. The much shorter delay \(\tau_{d,2}\) lets the rare class govern the earliest part of the survival curve despite its fivefold lower reactivation rate: direct evaluation of Eq.~\eqref{eq:trigger_prob_compact} gives probability \(\approx0.39\) that class~2 seeds the first detectable rebound. This inversion of the naive expectation---that the largest class dominates---is the central qualitative message of the heterogeneous model.

\begin{figure}[tbp]
\centering
\includegraphics[width=0.78\linewidth]{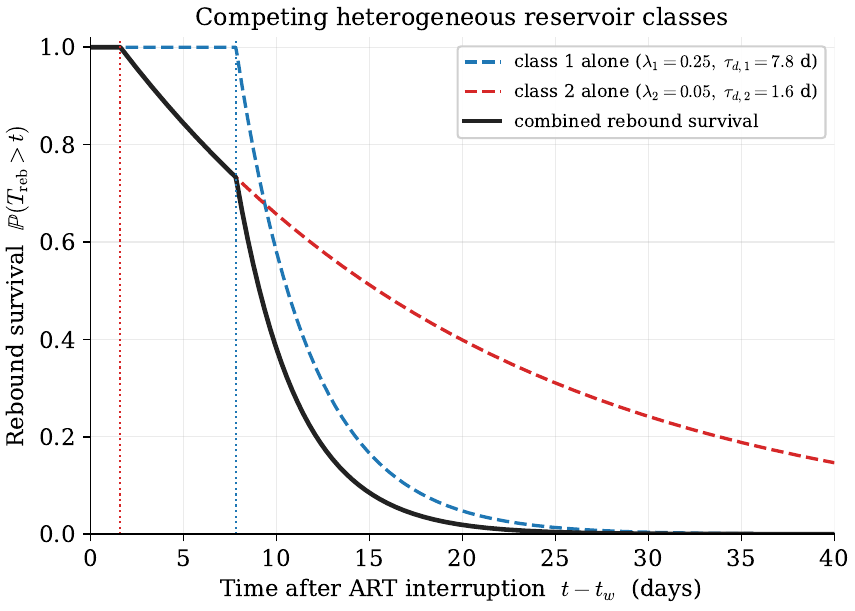}
\caption{Competing reservoir classes as shifted first-arrival processes, from the constant-rate survival law of Eq.~\eqref{eq:hetero_constant_survival_compact}. Class~1 (abundant, slow) and class~2 (rare, fast, high-output) are shown as dashed single-class survivals; their product (solid) is the combined rebound survival. A rare class with high founder output or rapid growth can seed detectable rebound earlier than a larger but slower class; here class~2 seeds first with probability \(\approx0.39\) from Eq.~\eqref{eq:trigger_prob_compact}.}
\label{fig:heterogeneous}
\end{figure}

\section{Validation, Inference, and Biological Interpretation}
\label{sec:validation_inference_interpretation}

A direct prediction of the first-passage formulation is that the observed rebound time depends logarithmically on the viral-load detection threshold. In the single-founder approximation,
\begin{equation}
T_{\rm reb}\approx T_1+\tau_e+\frac{1}{r}\log\left(\frac{V_{\rm det}}{v_0}\right),
\end{equation}
and, for constant successful-reactivation rate \(\lambda\), the median satisfies
\begin{equation}
Q_{0.5}(V_{\rm det})
=
C+\frac{1}{r}\log V_{\rm det},
\qquad
C=t_w+\tau_e+\frac{\log2}{\lambda}-\frac{1}{r}\log v_0 .
\label{eq:median_vs_threshold}
\end{equation}
Here \(\log V_{\rm det}\) is understood relative to a fixed reference unit, e.g. \(1\) copy mL\(^{-1}\). Thus, plotting median rebound time against \(\log V_{\rm det}\) should yield an approximately linear relation with slope \(1/r\).

The threshold-resolved medians reported by Gunst et al.~\cite{gunst2025ati} provide a useful consistency check. The median times to exceed \(50\), \(400\), and \(10{,}000\) copies mL\(^{-1}\) were approximately \(16\), \(21\), and \(32\) days. Linear regression against \(\log V_{\rm det}\) gives
\begin{equation}
\frac{1}{r}=3.05\ {\rm days},
\qquad
r\approx0.33\,{\rm day}^{-1},
\qquad
C\approx3.5\ {\rm days},
\label{eq:gunst_fit}
\end{equation}
with fitted medians \(15.5\), \(21.8\), and \(31.7\) days and \(R^2\simeq0.99\) (Figure~\ref{fig:gunst_fit}). This agreement supports the predicted logarithmic threshold dependence, but it should be interpreted as a consistency check rather than full parameter validation, since three cohort-level medians cannot separately identify \(\lambda\), \(r\), \(v_0\), and \(\tau_e\).

The calibrated value \(r\approx0.33\,{\rm day}^{-1}\) is lower than the nominal maximal early-growth value \(r=0.69\,{\rm day}^{-1}\) used in baseline examples. This is biologically plausible: the fitted \(r\) is an effective net growth rate over the clinically observed range from \(50\) to \(10{,}000\) copies mL\(^{-1}\), not the maximal expansion rate during the steepest phase of uncontrolled infection. Residual drug activity, target-cell limitation, innate immune pressure, and early establishment constraints can all reduce the effective growth rate. The intercept \(C\) constrains only the combination \(t_w+\tau_e+\log2/\lambda-(1/r)\log v_0\). For \(v_0=1\) and \(t_w=0\), \(C=\tau_e+\log2/\lambda\); thus \(\lambda\simeq0.20\,{\rm day}^{-1}\) leaves little additional eclipse delay, whereas \(\lambda=0.30\,{\rm day}^{-1}\) implies \(\tau_e\simeq1.2\) days. The dashed reference line in Figure~\ref{fig:gunst_fit} makes this contrast explicit: fixing the steeper maximal-growth value \(r=0.69\,{\rm day}^{-1}\) and anchoring it at the \(50\)-copy median produces threshold shifts that are too small, so the line falls well below the \(400\)- and \(10{,}000\)-copy medians. The observed logarithmic spacing therefore requires the slower effective rate \(r\approx0.33\,{\rm day}^{-1}\), which is the quantitative reason baseline maximal-growth parameters cannot reproduce the ATI medians.

\begin{figure}[tbp]
\centering
\includegraphics[width=0.72\linewidth]{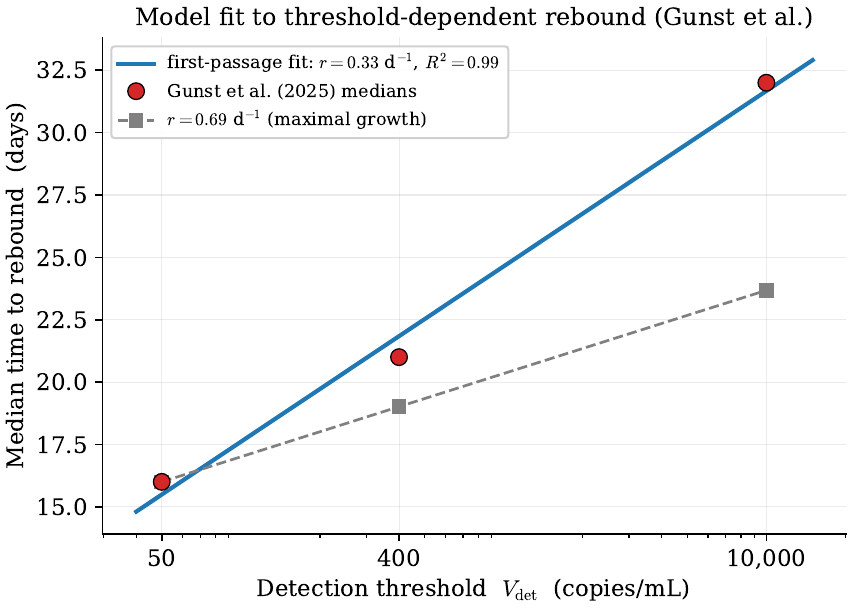}
\caption{Threshold-dependence consistency check using the Gunst et al.\ ATI meta-analysis~\cite{gunst2025ati}. Points show reported median times to exceed \(50\), \(400\), and \(10{,}000\) copies mL\(^{-1}\). The solid line is the first-passage prediction \(Q_{0.5}=C+r^{-1}\log V_{\rm det}\) of Eq.~\eqref{eq:median_vs_threshold}, whose least-squares fit is Eq.~\eqref{eq:gunst_fit}, giving \(r\approx0.33\,{\rm day}^{-1}\) and \(R^2\simeq0.99\). The dashed line shows the steeper baseline-growth case \(r=0.69\,{\rm day}^{-1}\), anchored at \(50\) copies mL\(^{-1}\), which underpredicts the observed threshold shifts.}
\label{fig:gunst_fit}
\end{figure}

The same formulation yields a likelihood for individual ATI data. For subject \(i\), define
\begin{equation}
\tau_{{\rm det},i}
=
\tau_{e,i}
+
\frac{1}{r_i}
\log\left(\frac{V_{{\rm det},i}}{v_{0i}}\right),
\end{equation}
and let \(\lambda_i(t;\theta)\) be the subject-specific successful-reactivation intensity. Introduce the shifted cumulative hazard
\begin{equation}
H_i(t;\theta)
=
\int_{t_{w,i}}^{t-\tau_{{\rm det},i}}
\lambda_i(s;\theta)\,ds,
\label{eq:subject_hazard}
\end{equation}
with the convention \(H_i(t;\theta)=0\) when \(t<t_{w,i}+\tau_{{\rm det},i}\). Then
\begin{equation}
S_i(t;\theta)=\exp[-H_i(t;\theta)]
\label{eq:subject_survival}
\end{equation}
is the rebound-survival probability for subject \(i\). For an exactly observed rebound time \(t_i\), the density is
\begin{equation}
f_i(t_i;\theta)
=
\lambda_i(t_i-\tau_{{\rm det},i};\theta)
\exp[-H_i(t_i;\theta)]
\mathbf{1}_{\{t_i\ge t_{w,i}+\tau_{{\rm det},i}\}} .
\label{eq:likelihood_uncensored}
\end{equation}
If rebound is interval-censored between the last below-threshold visit \(L_i\) and the first above-threshold visit \(R_i\), the likelihood contribution is
\begin{equation}
\Prob(L_i<T_{{\rm reb},i}\le R_i;\theta)
=
S_i(L_i;\theta)-S_i(R_i;\theta).
\label{eq:likelihood_interval}
\end{equation}
If no rebound is observed by \(c_i\), the contribution is \(S_i(c_i;\theta)\). Thus the full log-likelihood is
\begin{equation}
\ell(\theta)
=
\sum_{i\in\mathcal U}\log f_i(t_i;\theta)
+
\sum_{i\in\mathcal I}
\log\!\left[S_i(L_i;\theta)-S_i(R_i;\theta)\right]
+
\sum_{i\in\mathcal R}\log S_i(c_i;\theta),
\label{eq:full_loglik}
\end{equation}
where \(\mathcal U\), \(\mathcal I\), and \(\mathcal R\) denote exactly observed, interval-censored, and right-censored observations. This is the appropriate statistical object for ATI studies because it respects discrete sampling.

Post-treatment controllers can be represented by a defective survival model. If a fraction \(p_c\) does not generate an established rebound lineage during follow-up, then
\begin{equation}
S_i^{\rm pop}(t)
=
p_c+(1-p_c)S_i(t),
\qquad
f_i^{\rm pop}(t)=(1-p_c)f_i(t).
\label{eq:cure_fraction}
\end{equation}
A nonzero \(p_c\) produces a long-time survival plateau and may represent an empty rebound-competent reservoir, durable immune control, or failure of activated lineages to establish productive growth.

For the special case of constant \(\lambda_i(t)=\lambda\) and uncensored rebound times, the shifted observations
\begin{equation}
z_i=t_i-t_{w,i}-\tau_{{\rm det},i}>0
\end{equation}
satisfy an exponential likelihood,
\begin{equation}
L(\lambda)
=
\lambda^n\exp\left(-\lambda\sum_{i=1}^{n}z_i\right),
\end{equation}
so that
\begin{equation}
\widehat{\lambda}
=
\frac{n}{\sum_{i=1}^{n}z_i}.
\label{eq:mle_lambda}
\end{equation}
This estimator emphasizes that \(\lambda\) should be inferred after subtracting the deterministic detection delay; otherwise stochastic activation waiting and growth-to-detection are confounded. Heterogeneity may be included through
\begin{equation}
\log\lambda_i=\mu_\lambda+\eta_i,
\qquad
\eta_i\sim N(0,\sigma_\lambda^2),
\end{equation}
or by linking reactivation to reservoir measurements,
\begin{equation}
\lambda_i=\rho R_i,
\end{equation}
where \(R_i\) is an effective rebound-competent reservoir measure and \(\rho\) is the successful-reactivation rate per reservoir unit.

The biological interpretation follows from
\begin{equation}
T_{\rm reb}
\approx
\underbrace{T_1}_{\text{successful-reactivation waiting time}}
+
\underbrace{\tau_e+
\frac{1}{r}\log\left(\frac{V_{\rm det}}{v_0}\right)}_{\text{eclipse and growth-to-detection delay}} .
\label{eq:bio_interpretation}
\end{equation}
Reservoir reduction lowers \(\lambda\), thereby extending the stochastic waiting time. Residual ART suppresses early \(\lambda(t)\). Immune control can reduce establishment probability, lowering \(\lambda\), and can also reduce \(r\), increasing the growth-to-detection delay. Latency-reversing agents increase raw activation, but their clinical benefit depends on whether establishment and post-activation growth are simultaneously controlled. Assay sensitivity changes the observed rebound time logarithmically, even if the underlying biological process is unchanged.

Together, these results separate the clinically observed first-passage time from ensemble-average viral dynamics. The single-founder approximation agrees with stochastic simulation when early rebound is dominated by the first established lineage; mean-field threshold crossing underestimates typical rebound in rare-reactivation regimes; slow ART washout delays rebound; and reservoir heterogeneity can reorder the source of first detectable virus. The formulation therefore provides both an analytical interpretation of rebound timing and a likelihood framework for ATI data.

\section{Summary and Conclusions}
\label{sec:summary_conclusion}

This work develops a stochastic first-passage theory for HIV rebound after ART interruption. The central point is that clinical rebound is not the hidden activation time of a latent cell and not the threshold crossing of an ensemble-mean viral load. It is the first time at which an individual stochastic viral-load trajectory crosses an assay threshold:
\begin{equation}
T_{\rm reb}
=
\inf\{t\ge t_w:V(t)\ge V_{\rm det}\}.
\end{equation}

Successful established reactivation events generate a Poisson shot-noise viral-load process,
\begin{equation}
V(t)=v_0\sum_{T_i\le t}e^{r(t-T_i)},
\end{equation}
with Laplace transform
\begin{equation}
\Phi_t(q)
=
\exp\left\{
-\int_{t_w}^{t}
\lambda(s)
\left[
1-\exp\left(-qv_0e^{r(t-s)}\right)
\right]ds
\right\}.
\end{equation}
For \(r>0\), \(V(t)\) is monotone after establishment, and
\begin{equation}
\Prob(T_{\rm reb}>t)=F_{V(t)}(V_{\rm det}^{-}),
\end{equation}
which gives an exact semi-analytical rebound-time distribution by Laplace inversion.

The most transparent closed form arises in the single-founder regime:
\begin{equation}
T_{\rm reb}
\approx
T_1+\tau_{\rm det},
\qquad
\tau_{\rm det}
=
\tau_e+
\frac{1}{r}
\log\left(\frac{V_{\rm det}}{v_0}\right).
\end{equation}
Thus, rebound timing separates into a stochastic waiting time for successful reactivation and a deterministic delay from activation to detectability. For a general time-dependent successful-reactivation intensity,
\begin{equation}
\Prob(T_{\rm reb}>t)
\approx
\exp\left[
-\int_{t_w}^{t-\tau_{\rm det}}\lambda(s)\,ds
\right],
\end{equation}
with the integral set to zero when \(t<t_w+\tau_{\rm det}\).

This decomposition gives direct biological interpretation. Reservoir-reduction interventions lower \(\lambda\). Immune control can lower \(\lambda\) by reducing establishment probability and can increase \(\tau_{\rm det}\) by reducing \(r\). ART washout delays rebound by suppressing early \(\lambda(t)\). Periodic immune activation introduces phase-dependent shifts, Cox-process activation captures immune-driven overdispersion, and heterogeneous reservoirs generate competing first-arrival processes. A small reservoir class can dominate early rebound if it has sufficiently high establishment probability, founder output, or growth rate.

The theory also explains why mean-field viral-load predictions can be misleading. In rare-reactivation regimes, \(\mathbb E[V(t)]\) may be dominated by a minority of early high-growth trajectories, causing the mean-field threshold crossing to occur earlier than the median or typical observed rebound time. First-passage survival functions, quantiles, and interval-censored likelihoods are therefore the appropriate objects for ATI analysis.

Finally, threshold-resolved ATI medians support the predicted logarithmic dependence of rebound time on \(V_{\rm det}\). Fitting the reported \(50\), \(400\), and \(10{,}000\) copies mL\(^{-1}\) medians yields an effective growth rate \(r\approx0.33\,{\rm day}^{-1}\), lower than the nominal maximal early-growth value. This suggests that clinically observed rebound reflects effective growth shaped by residual drug activity, target-cell limitation, immune pressure, and early establishment constraints. The first-passage framework therefore complements earlier activation-survival and mean-load models by focusing directly on the observable measured in ATI trials: the first detectable viral rebound.

\appendix
\renewcommand{\thesection}{\Alph{section}}
\titleformat{\section}
  {\normalfont\Large\bfseries}{Appendix~\thesection.}{0.6em}{}
\titleformat{\subsection}
  {\normalfont\large\bfseries}{\thesubsection}{0.6em}{}

\section{Poisson Reactivation and Shot-Noise Viral Load}
\label{app:poisson_waiting_detailed}
\label{app:shot_noise_detailed}

Let \(N(t)\) count successful established reactivation events after treatment interruption or effective washout at \(\tw\). For an inhomogeneous Poisson process with intensity \(\lambda(t)\), the count probabilities
\[
P(n,t)=\Prob\{N(t)-N(\tw)=n\}
\]
satisfy
\begin{equation}
\frac{\dd P(n,t)}{\dd t}
=
\lambda(t)\left[P(n-1,t)-P(n,t)\right],
\qquad n\ge1,
\label{eq:app_master}
\end{equation}
with \(\dot P(0,t)=-\lambda(t)P(0,t)\) and \(P(n,\tw)=\delta_{n0}\). Introducing
\begin{equation}
\Lambda(t)=\int_{\tw}^{t}\lambda(s)\dd s,
\end{equation}
the probability-generating function \(G(z,t)=\sum_{n\ge0}P(n,t)z^n\) obeys
\[
\partial_tG=\lambda(t)(z-1)G,
\qquad
G(z,\tw)=1,
\]
and hence
\[
G(z,t)=\exp[(z-1)\Lambda(t)].
\]
Therefore
\begin{equation}
P(n,t)
=
\frac{\Lambda(t)^n}{n!}\e^{-\Lambda(t)}.
\label{eq:app_poisson_solution}
\end{equation}
In particular, the first successful reactivation time \(T_1\) has survival and density
\begin{equation}
\Prob(T_1>t)=\e^{-\Lambda(t)},
\qquad
f_{T_1}(t)=\lambda(t)\e^{-\Lambda(t)}.
\label{eq:app_first_waiting}
\end{equation}
For constant \(\lambda(t)=\lambda\), \(T_1-\tw\sim{\rm Exp}(\lambda)\).

Each successful event at time \(T_i\) seeds a founder contribution \(v_0\). Between events, the viral load grows with net rate \(r\), so
\begin{equation}
\dd V(t)=rV(t)\dd t+v_0N(\dd t),
\qquad V(\tw)=0.
\label{eq:app_jump_growth}
\end{equation}
Multiplication by \(\e^{-rt}\) and integration yield
\begin{equation}
V(t)
=
\int_{\tw}^{t}v_0\e^{r(t-s)}N(\dd s)
=
v_0\sum_{T_i\le t}\e^{r(t-T_i)}.
\label{eq:app_shot_noise}
\end{equation}
Thus, post-interruption viral load is a Poisson shot-noise process: a random superposition of exponentially growing founder lineages.

\section{Laplace Transform and Cumulants}
\label{app:laplace_discretization}
\label{app:cumulants_detailed_extra}
\label{app:cumulants}

Because \(V(t)\) is a linear functional of a Poisson random measure, its Laplace transform follows from the Poisson functional identity. For \(q\ge0\),
\begin{equation}
\Phi_t(q)
=
\E[\e^{-qV(t)}]
=
\exp\left\{
-\int_{\tw}^{t}\lambda(s)
\left[
1-\exp\!\left(-qv_0\e^{r(t-s)}\right)
\right]\dd s
\right\}.
\label{eq:app_laplace}
\end{equation}
Equivalently, this formula is obtained by partitioning time, assigning to each interval the independent Poisson increment \(\Delta N_k\), and taking the continuum limit of
\[
\E[\e^{-qV(t)}]
\approx
\prod_k
\exp\left[
-\lambda(s_k)\Delta s
\left(1-\e^{-qv_0e^{r(t-s_k)}}\right)
\right].
\]

The cumulants follow by expanding \(\log\Phi_t(q)\). Since
\[
1-\e^{-qz}
=
\sum_{m=1}^{\infty}(-1)^{m+1}\frac{q^mz^m}{m!},
\]
comparison with the Laplace-cumulant expansion
\[
\log\E[\e^{-qV(t)}]
=
\sum_{m=1}^{\infty}(-1)^m\frac{q^m}{m!}\kappa_m(t)
\]
gives
\begin{equation}
\kappa_m(t)
=
v_0^m
\int_{\tw}^{t}\lambda(s)\e^{mr(t-s)}\dd s,
\qquad m=1,2,\ldots .
\label{eq:app_cumulants}
\end{equation}
Hence
\begin{equation}
\E[V(t)]
=
v_0\int_{\tw}^{t}\lambda(s)\e^{r(t-s)}\dd s,
\qquad
\Var[V(t)]
=
v_0^2\int_{\tw}^{t}\lambda(s)\e^{2r(t-s)}\dd s.
\label{eq:app_mean_var}
\end{equation}
For constant \(\lambda\), with \(\Delta=t-\tw\),
\begin{equation}
\kappa_m(t)
=
\frac{v_0^m\lambda}{mr}
\left(\e^{mr\Delta}-1\right).
\label{eq:app_const_cumulants}
\end{equation}

\section{First Passage, Single-Founder Bound, and Mean-Field Bias}
\label{app:first_passage_bounds}
\label{app:meanfield_asymptotics}
\label{app:constant_rate_detailed}

When \(r>0\), \(V(t)\) is nondecreasing: it grows between events and jumps upward at each successful reactivation. Therefore
\begin{equation}
\{\Treb\le t\}=\{V(t)\ge\Vdet\},
\end{equation}
and the exact monotone first-passage relation is
\begin{equation}
\Prob(\Treb>t)
=
\Prob(V(t)<\Vdet)
=
F_{V(t)}(\Vdet^-).
\label{eq:app_exact_fp}
\end{equation}
Together with the Laplace transform in Eq.~\eqref{eq:app_laplace}, this gives an exact semi-analytical rebound-time law.

A founder seeded at \(s\) reaches the detection threshold by time \(t\) if
\[
v_0\e^{r(t-s)}\ge\Vdet.
\]
Thus
\begin{equation}
s\le t-\td,
\qquad
\td=\frac{1}{r}\log\left(\frac{\Vdet}{v_0}\right).
\label{eq:app_td}
\end{equation}
The occurrence of at least one successful event before \(t-\td\) is sufficient for rebound by \(t\). Hence
\begin{equation}
\{\Treb>t\}
\subseteq
\{N(t-\td)-N(\tw)=0\},
\end{equation}
and
\begin{equation}
\Prob(\Treb>t)
\le
\exp\left[-\int_{\tw}^{t-\td}\lambda(s)\dd s\right].
\label{eq:app_sf_bound}
\end{equation}
The single-founder approximation replaces this upper bound by equality; the error arises from cooperative crossing by multiple sub-detectable lineages. For approximately constant \(\lambda\), a rough scale for this correction is
\begin{equation}
\Prob\{N(\td)\ge2\}
=
1-\e^{-\lambda\td}(1+\lambda\td)
=
\frac{(\lambda\td)^2}{2}+O((\lambda\td)^3).
\label{eq:app_coop_scale}
\end{equation}
The eclipse phase is restored by \(\td\to\tdelay=\tecl+\td\).

For constant \(\lambda\), the single-founder rebound time is the shifted exponential
\begin{equation}
\Treb=\tw+\td+E,
\qquad
E\sim{\rm Exp}(\lambda).
\end{equation}
Therefore
\begin{equation}
S_{\rm reb}(t)
=
\begin{cases}
1, & t<\tw+\td,\\[4pt]
\e^{-\lambda(t-\tw-\td)}, & t\ge\tw+\td,
\end{cases}
\label{eq:app_const_survival}
\end{equation}
with density \(f_{\rm reb}(t)=\lambda e^{-\lambda(t-\tw-\td)}\), hazard \(h_{\rm reb}(t)=\lambda\), and quantile
\begin{equation}
Q_p
=
\tw+\td-\frac{1}{\lambda}\log(1-p).
\label{eq:app_const_quantile}
\end{equation}

The mean-field threshold time \(\bar T\), defined by \(\E[V(\bar T)]=\Vdet\), is not a first-passage statistic. For constant \(\lambda\),
\begin{equation}
\bar T
=
\tw+
\frac{1}{r}
\log\left(1+\frac{r\Vdet}{v_0\lambda}\right).
\label{eq:app_meanfield}
\end{equation}
As \(\lambda\to0\),
\[
\bar T-\tw\sim r^{-1}\log(1/\lambda),
\qquad
\E[\Treb]-\tw\sim \lambda^{-1}.
\]
Thus, in the rare-reactivation regime, mean-field crossing can occur much earlier than typical observed rebound.

\section{Explicit Time-Dependent Intensities}
\label{app:washout_detailed}
\label{app:periodic_detailed}
\label{app:cox_detailed}

\subsection*{ART washout}

Let
\begin{equation}
\lambda(t)=\lambda_0\left[1-A_0\e^{-k(t-\tw)}\right],
\qquad
0\le A_0\le1,\quad k=k_{\rm drug}.
\label{eq:app_washout_lambda}
\end{equation}
With \(y=t-\tw-\td\), the shifted cumulative hazard is
\begin{equation}
H_{\rm w}(y)
=
\lambda_0\left[
y-\frac{A_0}{k}\left(1-\e^{-ky}\right)
\right],
\qquad y\ge0.
\label{eq:app_washout_hazard}
\end{equation}
Hence
\begin{equation}
S_{\rm reb}(t)\approx \e^{-H_{\rm w}(y)},
\qquad
f_{\rm reb}(t)
\approx
\lambda_0(1-A_0\e^{-ky})\e^{-H_{\rm w}(y)}\mathbf{1}_{\{y\ge0\}}.
\label{eq:app_washout_surv_density}
\end{equation}
For \(A_0=1\) and \(ky\ll1\),
\[
H_{\rm w}(y)
=
\frac{1}{2}\lambda_0ky^2+O(y^3),
\]
so the early waiting time is approximately Rayleigh:
\begin{equation}
S_{\rm reb}(t)
\approx
\exp\left[-\frac{1}{2}\lambda_0ky^2\right],
\qquad
\E[y]\approx
\sqrt{\frac{\pi}{2\lambda_0k}}.
\label{eq:app_rayleigh}
\end{equation}

The mean waiting contribution is
\begin{equation}
\int_0^\infty \e^{-H_{\rm w}(y)}\dd y
=
e^{c}
\sum_{n=0}^{\infty}
\frac{(-c)^n}{n!}
\frac{1}{\lambda_0+nk},
\qquad
c=\frac{\lambda_0A_0}{k}.
\label{eq:app_washout_mean_series}
\end{equation}
This recovers \(1/\lambda_0\) as \(A_0\to0\) or \(k\to\infty\).

The quantile equation \(H_{\rm w}(y_p)=z_p\), where \(z_p=-\log(1-p)\), can be solved using the Lambert function. For \(A_0>0\),
\begin{equation}
y_p
=
-\frac{1}{k}
\log\left[
-\frac{1}{A_0}
W_0\!\left(
-A_0\exp\left[-A_0-\frac{kz_p}{\lambda_0}\right]
\right)
\right],
\qquad
Q_p=\tw+\td+y_p.
\label{eq:app_washout_quantile}
\end{equation}
For \(A_0=0\), this reduces to \(y_p=z_p/\lambda_0\).

\subsection*{Periodic immune modulation}

Let
\begin{equation}
\lambda(t)=
\lambda_0\left[
1+a\sin(\omega(t-\tw)+\phi)
\right],
\qquad |a|<1.
\label{eq:app_periodic_lambda}
\end{equation}
For \(y=t-\tw-\td\),
\begin{equation}
H_{\rm p}(y)
=
\lambda_0y+
\frac{\lambda_0a}{\omega}
\left[
\cos\phi-\cos(\omega y+\phi)
\right].
\label{eq:app_periodic_hazard}
\end{equation}
Thus
\begin{equation}
S_{\rm reb}(t)\approx\e^{-H_{\rm p}(y)},
\qquad
f_{\rm reb}(t)\approx
\lambda_0[1+a\sin(\omega y+\phi)]\e^{-H_{\rm p}(y)}
\mathbf{1}_{\{y\ge0\}}.
\label{eq:app_periodic_surv_density}
\end{equation}

For small \(|a|\), the quantile admits the expansion
\begin{equation}
Q_p
=
\tw+\td+\frac{z_p}{\lambda_0}
-
\frac{a}{\omega}
\left[
\cos\phi-
\cos\left(\frac{\omega z_p}{\lambda_0}+\phi\right)
\right]
+O(a^2).
\label{eq:app_periodic_quantile}
\end{equation}
The mean waiting contribution is
\begin{equation}
\int_0^\infty S_{\rm reb}(y)\dd y
=
\frac{1}{\lambda_0}
-\frac{a}{\omega}\cos\phi
+
\frac{a\lambda_0}{\omega}
\frac{\lambda_0\cos\phi-\omega\sin\phi}
{\lambda_0^2+\omega^2}
+O(a^2).
\label{eq:app_periodic_mean}
\end{equation}
This shows explicitly how immune phase at interruption advances or delays rebound.

\subsection*{Cox-process activation}

If the reactivation intensity is itself random, then, conditional on a realization of \(\lambda(t)\),
\begin{equation}
S_{\rm reb}(t\mid\lambda)
=
\exp[-\Xi_y],
\qquad
\Xi_y=\int_{\tw}^{\tw+y}\lambda(s)\dd s,
\qquad
y=t-\tw-\td.
\end{equation}
The unconditional survival is therefore
\begin{equation}
S_{\rm reb}(t)=\E[\e^{-\Xi_y}],
\label{eq:app_cox_survival}
\end{equation}
the Laplace transform of the integrated random hazard. If \(\kappa_n^{(\Xi)}(y)\) are the cumulants of \(\Xi_y\), then
\begin{equation}
\log S_{\rm reb}(t)
=
\sum_{n=1}^{\infty}
\frac{(-1)^n}{n!}\kappa_n^{(\Xi)}(y).
\end{equation}
Keeping two cumulants gives
\begin{equation}
S_{\rm reb}(t)
\approx
\exp\left[
-\E[\Xi_y]+\frac{1}{2}\Var(\Xi_y)
\right].
\label{eq:app_cox_cumulant}
\end{equation}
Thus, hazard variability increases survival relative to a deterministic process with the same mean integrated hazard.

For short-correlated fluctuations with
\[
\E[\Xi_y]=\lambda_0y,
\qquad
\Var(\Xi_y)\approx Dy,
\]
one obtains
\begin{equation}
S_{\rm reb}(t)
\approx
\exp[-(\lambda_0-D/2)y],
\qquad
\E[\Treb]
\approx
\tw+\td+\frac{1}{\lambda_0-D/2},
\label{eq:app_cox_mean}
\end{equation}
valid when higher cumulants are negligible and \(D<2\lambda_0\).

A nonnegative intensity may be specified by the lognormal Cox form
\begin{equation}
\lambda(t)=
\lambda_0\exp\left[
\sigma X(t)-\frac{1}{2}\sigma^2
\right],
\end{equation}
where \(X(t)\) has zero mean and unit marginal variance, ensuring \(\E[\lambda(t)]=\lambda_0\).

\section{Reservoir Heterogeneity and Establishment}
\label{app:hetero_competing_detailed}
\label{app:branching_detailed}
\label{app:hetero_mean}

For reservoir class \(j\), let \(\lambda_j(t)\), \(v_j\), and \(r_j\) denote the successful-reactivation intensity, founder output, and net growth rate. The viral load is
\begin{equation}
V(t)
=
\sum_{j=1}^{J}
\sum_{T_{ij}\le t}
v_j\e^{r_j(t-T_{ij})}.
\label{eq:app_hetero_V}
\end{equation}
The corresponding Laplace transform is
\begin{equation}
\Phi_t(q)
=
\exp\left\{
-\sum_{j=1}^{J}
\int_{\tw}^{t}
\lambda_j(s)
\left[
1-\exp\left(-qv_j\e^{r_j(t-s)}\right)
\right]\dd s
\right\}.
\label{eq:app_hetero_laplace}
\end{equation}

In the single-founder approximation, class \(j\) has detection delay
\begin{equation}
\tau_{{\rm d},j}
=
\frac{1}{r_j}
\log\left(\frac{\Vdet}{v_j}\right).
\end{equation}
The class-specific detectable arrival time is \(R_j=T_{1j}+\tau_{{\rm d},j}\), and
\[
\Treb=\min_jR_j.
\]
Assuming independent class-specific reactivation processes,
\begin{equation}
\Prob(\Treb>t)
\approx
\exp\left[
-\sum_{j=1}^{J}
\int_{\tw}^{t-\tau_{{\rm d},j}}
\lambda_j(s)\dd s
\right],
\label{eq:app_hetero_survival}
\end{equation}
where an integral is zero if its upper limit is below \(\tw\). The density is
\begin{equation}
f_{\rm reb}(t)
=
h_{\rm tot}(t)S_{\rm reb}(t),
\qquad
h_{\rm tot}(t)
=
\sum_{j=1}^{J}
\lambda_j(t-\tau_{{\rm d},j})
\mathbf{1}_{\{t\ge\tw+\tau_{{\rm d},j}\}}.
\label{eq:app_hetero_density}
\end{equation}
The probability that class \(j\) seeds detectable rebound is
\begin{equation}
\Prob(J^*=j)
=
\int_{\tw}^{\infty}
\lambda_j(t-\tau_{{\rm d},j})
\mathbf{1}_{\{t\ge\tw+\tau_{{\rm d},j}\}}
S_{\rm reb}(t)\dd t.
\label{eq:app_trigger_prob}
\end{equation}

For constant class-specific rates, the survival becomes
\begin{equation}
S(t)=
\exp\left[
-\sum_{j=1}^{J}
\lambda_j(t-\tw-\tau_{{\rm d},j})_+
\right].
\label{eq:app_hetero_const_survival}
\end{equation}
Ordering the delays as
\[
\tau_{(1)}\le\cdots\le\tau_{(J)}
\]
with corresponding rates \(\lambda_{(j)}\), define
\[
\Lambda_k=\sum_{m=1}^{k}\lambda_{(m)},
\qquad
B_k=\sum_{m=1}^{k}\lambda_{(m)}\tau_{(m)}.
\]
Then
\begin{equation}
\E[\Treb]
=
\tw+\tau_{(1)}
+
\sum_{k=1}^{J-1}
e^{B_k}
\frac{
e^{-\Lambda_k\tau_{(k)}}-
e^{-\Lambda_k\tau_{(k+1)}}
}{\Lambda_k}
+
e^{B_J}
\frac{e^{-\Lambda_J\tau_{(J)}}}{\Lambda_J}.
\label{eq:app_hetero_mean_explicit}
\end{equation}
This finite sum gives the exact single-founder mean for constant heterogeneous classes.

The distinction between raw activation and successful establishment is captured by a birth-death branching process. If an activated lineage gives birth at rate \(b\) and dies at rate \(d\), the extinction probability \(q\) from one founder satisfies
\[
q=\frac{d}{b+d}+\frac{b}{b+d}q^2.
\]
Hence \(q=1\) or \(q=d/b\), and the establishment probability is
\begin{equation}
p_{\rm est}
=
\begin{cases}
1-d/b, & b>d,\\
0, & b\le d.
\end{cases}
\label{eq:app_establishment}
\end{equation}
Therefore the successful-reactivation intensity used in the first-passage theory is
\begin{equation}
\lambda(t)=p_{\rm est}(t)\lambda_{\rm act}(t).
\label{eq:app_successful_lambda}
\end{equation}

\section{Time-Dependent Growth and Cooperative Corrections}
\label{app:time_dependent_growth_detailed}
\label{app:two_founder}

The constant-growth assumption can be relaxed. A lineage seeded at \(s\) under a time-dependent net growth rate \(r(t)\) satisfies
\begin{equation}
v(u;s)
=
v_0
\exp\left[\int_s^u r(a)\dd a\right],
\qquad u\ge s.
\label{eq:app_td_growth_kernel}
\end{equation}
The total viral load is
\begin{equation}
V(t)
=
\int_{\tw}^{t}
v_0
\exp\left[\int_s^t r(a)\dd a\right]
N(\dd s).
\end{equation}
If \(r(t)\ge0\), detectability by time \(t\) is determined by the final value \(v(t;s)\). For general time-dependent growth, however, the correct first-passage condition must use the running maximum. A founder seeded at \(s\) becomes detectable by time \(t\) if
\begin{equation}
\sup_{u\in[s,t]}
\int_s^u r(a)\dd a
\ge
\log\left(\frac{\Vdet}{v_0}\right).
\label{eq:app_general_detectability}
\end{equation}
Therefore, the single-founder survival generalizes to
\begin{equation}
\Prob(\Treb>t)
\approx
\exp\left[
-\int_{\tw}^{t}
\lambda(s)
\mathbf{1}
\left\{
\sup_{u\in[s,t]}
\int_s^u r(a)\dd a
\ge
\log\left(\frac{\Vdet}{v_0}\right)
\right\}
\dd s
\right].
\label{eq:app_time_dependent_survival}
\end{equation}
When \(r(t)=r>0\), the condition reduces to \(s\le t-\td\), recovering the shifted-hazard formula.

The leading cooperative correction arises when a second lineage contributes before the first founder alone reaches threshold. If \(T_1<T_2\), then for \(t\ge T_2\)
\begin{equation}
V(t)
=
v_0e^{r(t-T_1)}
+
v_0e^{r(t-T_2)}
=
v_0e^{r(t-T_2)}
\left(1+e^{r(T_2-T_1)}\right).
\end{equation}
The two-founder crossing time, conditional on both lineages being present, is therefore
\begin{equation}
t_c^{(2)}
=
T_2+
\frac{1}{r}
\log\left[
\frac{\Vdet}
{v_0\left(1+e^{r(T_2-T_1)}\right)}
\right],
\label{eq:app_two_founder_exact}
\end{equation}
provided this occurs before the first-founder crossing \(T_1+\td\). In the equal-age limiting case, \(n\) comparable founders reduce the delay to
\begin{equation}
\td^{(n)}
=
\frac{1}{r}
\log\left(\frac{\Vdet}{nv_0}\right).
\label{eq:app_n_founder_delay}
\end{equation}
These corrections become relevant when multiple successful events occur within one detection-delay window, i.e. when \(\lambda\td\) is not small.


\begin{thebibliography}{999}

\bibitem{Taye2025CM}
M. A. Taye, ``Stochastic modeling of HIV reactivation under ART washout and immune fluctuations,'' \emph{Contemporary Mathematics} \textbf{6}(5), 5708--5739 (2025). \url{https://doi.org/10.37256/cm.6520256801}.

\bibitem{Perelson1993}
A. S. Perelson, D. E. Kirschner, and R. De Boer, ``Dynamics of HIV infection of CD4+ T cells,'' \emph{Mathematical Biosciences} \textbf{114}, 81--125 (1993).

\bibitem{Perelson1996}
A. S. Perelson, A. U. Neumann, M. Markowitz, J. M. Leonard, and D. D. Ho, ``HIV-1 dynamics in vivo: virion clearance rate, infected cell life-span, and viral generation time,'' \emph{Science} \textbf{271}, 1582--1586 (1996).

\bibitem{Ho1995}
D. D. Ho, A. U. Neumann, A. S. Perelson, W. Chen, J. M. Leonard, and M. Markowitz, ``Rapid turnover of plasma virions and CD4 lymphocytes in HIV-1 infection,'' \emph{Nature} \textbf{373}, 123--126 (1995).

\bibitem{Wei1995}
X. Wei, S. K. Ghosh, M. E. Taylor, V. A. Johnson, E. A. Emini, P. Deutsch, et al., ``Viral dynamics in human immunodeficiency virus type 1 infection,'' \emph{Nature} \textbf{373}, 117--122 (1995).

\bibitem{Chun1997}
T. W. Chun, L. Stuyver, S. B. Mizell, L. A. Ehler, J. A. Mican, M. Baseler, et al., ``Presence of an inducible HIV-1 latent reservoir during highly active antiretroviral therapy,'' \emph{Proceedings of the National Academy of Sciences USA} \textbf{94}, 13193--13197 (1997).

\bibitem{Siliciano2003}
J. D. Siliciano, J. Kajdas, D. Finzi, T. C. Quinn, K. Chadwick, J. B. Margolick, et al., ``Long-term follow-up studies confirm the stability of the latent reservoir for HIV-1 in resting CD4+ T cells,'' \emph{Nature Medicine} \textbf{9}, 727--728 (2003).

\bibitem{Murray2016}
A. J. Murray, K. J. Kwon, D. L. Farber, and R. F. Siliciano, ``The latent reservoir for HIV-1: how immunologic memory and clonal expansion contribute to HIV-1 persistence,'' \emph{Journal of Immunology} \textbf{197}, 407--417 (2016).

\bibitem{Li2020}
J. Z. Li, et al., ``Proviruses with identical sequences comprise a large fraction of the replication-competent HIV reservoir,'' \emph{Proceedings of the National Academy of Sciences USA} \textbf{117}, 3886--3893 (2020).

\bibitem{Hill2014}
A. L. Hill, D. I. S. Rosenbloom, F. Fu, M. A. Nowak, and R. F. Siliciano, ``Predicting the outcomes of treatment to eradicate the latent reservoir for HIV-1,'' \emph{Proceedings of the National Academy of Sciences USA} \textbf{111}, 13475--13480 (2014).

\bibitem{Hill2018}
A. L. Hill, D. I. S. Rosenbloom, E. Goldberg, E. Hanhauser, D. R. Kuritzkes, R. F. Siliciano, et al., ``Real-time predictions of reservoir size and rebound time during antiretroviral therapy interruption trials for HIV,'' \emph{PLoS Pathogens} \textbf{14}, e1007333 (2018).

\bibitem{pinkevych2015latency}
M. Pinkevych, D. Cromer, M. Tolstrup, A. J. Grimm, D. A. Cooper, S. R. Lewin, O. S. S{\o}gaard, T. A. Rasmussen, S. J. Kent, A. D. Kelleher, and M. P. Davenport, ``HIV reactivation from latency after treatment interruption occurs on average every 5--8 days: implications for HIV remission,'' \emph{PLoS Pathogens} \textbf{11}(7), e1005000 (2015). doi:10.1371/journal.ppat.1005000.

\bibitem{conway2019rebound}
J. M. Conway, A. S. Perelson, and J. Z. Li, ``Predictions of time to HIV viral rebound following ART suspension that incorporate personal biomarkers,'' \emph{PLoS Computational Biology} \textbf{15}(7), e1007229 (2019). doi:10.1371/journal.pcbi.1007229.

\bibitem{NowakBangham1996}
M. A. Nowak and C. R. M. Bangham, ``Population dynamics of immune responses to persistent viruses,'' \emph{Science} \textbf{272}, 74--79 (1996).

\bibitem{Perelson2002}
A. S. Perelson, ``Modelling viral and immune system dynamics,'' \emph{Nature Reviews Immunology} \textbf{2}, 28--36 (2002).

\bibitem{Wodarz2002}
D. Wodarz and M. A. Nowak, ``Mathematical models of HIV pathogenesis and treatment,'' \emph{BioEssays} \textbf{24}, 1178--1187 (2002).

\bibitem{gunst2025ati}
J. D. Gunst, J. Gohil, J. Z. Li, R. J. Bosch, et al., ``Time to HIV viral rebound and frequency of post-treatment control after analytical interruption of antiretroviral therapy: an individual data-based meta-analysis of 24 prospective studies,'' \emph{Nature Communications} \textbf{16}, 906 (2025). doi:10.1038/s41467-025-56116-1.

\bibitem{fennessey2017barcoded}
C. M. Fennessey, M. Pinkevych, T. T. Immonen, A. Reynaldi, V. Venturi, P. Nadella, C. Reid, L. Newman, L. Lipkey, K. Oswald, et al., ``Genetically-barcoded SIV facilitates enumeration of rebound variants and estimation of reactivation rates in nonhuman primates following interruption of suppressive antiretroviral therapy,'' \emph{PLoS Pathogens} \textbf{13}(5), e1006359 (2017). doi:10.1371/journal.ppat.1006359.

\bibitem{VanDorp2020}
C. H. Van Dorp, J. M. Conway, D. H. Barouch, J. B. Whitney, and A. S. Perelson, ``Models of SIV rebound after treatment interruption that involve multiple reactivation events,'' \emph{PLoS Computational Biology} \textbf{16}, e1008241 (2020).

\bibitem{li2016reservoir}
J. Z. Li, B. Etemad, H. Ahmed, E. Aga, R. J. Bosch, J. W. Mellors, D. R. Kuritzkes, M. M. Lederman, M. Para, and R. T. Gandhi, ``The size of the expressed HIV reservoir predicts timing of viral rebound after treatment interruption,'' \emph{AIDS} \textbf{30}(3), 343--353 (2016). doi:10.1097/QAD.0000000000000953.

\bibitem{pasternak2020carna}
A. O. Pasternak, S. DeMaster, et al., ``Cell-associated HIV-1 RNA predicts viral rebound and disease progression after discontinuation of temporary early ART,'' \emph{JCI Insight} \textbf{5}(6), e134196 (2020). doi:10.1172/jci.insight.134196.

\bibitem{sneller2020kinetics}
M. C. Sneller, E. D. Huiting, K. E. Clarridge, C. Seamon, J. Blazkova, J. S. Justement, et al., ``Kinetics of plasma HIV rebound in the era of modern antiretroviral therapy,'' \emph{Journal of Infectious Diseases} \textbf{222}(10), 1655--1659 (2020). doi:10.1093/infdis/jiaa270.

\bibitem{Pinkevych2019}
M. Pinkevych, D. Cromer, M. P. Davenport, et al., ``Modeling of experimental data supports HIV reactivation from latency after treatment interruption at high but variable rates,'' \emph{eLife} \textbf{8}, e49022 (2019).

\bibitem{Wu2020}
Y. Wu, M. Pinkevych, Z. Xu, B. F. Keele, M. P. Davenport, and D. Cromer, ``Impact of fluctuation in frequency of human immunodeficiency virus/simian immunodeficiency virus reactivation during antiretroviral therapy interruption,'' \emph{Proceedings of the Royal Society B} \textbf{287}, 20200354 (2020).

\bibitem{Kingman1993}
J. F. C. Kingman, \emph{Poisson Processes} (Oxford University Press, Oxford, 1993).

\bibitem{Feller1971}
W. Feller, \emph{An Introduction to Probability Theory and Its Applications, Vol. II}, 2nd ed. (Wiley, New York, 1971).

\bibitem{CoxMiller1965}
D. R. Cox and H. D. Miller, \emph{The Theory of Stochastic Processes} (Chapman and Hall, London, 1965).


\bibitem{Taye2026StochasticFirstPassage}
M. A. Taye, \textit{Stochastic First-Passage Theory of HIV Viral Rebound Following Latent Reservoir Reactivation}, arXiv:2607.04910 (2026).

\bibitem{Taye2026BiologicalTimeEvolutionary}
M. A. Taye, \textit{Biological Time, Evolutionary Optimization, and Gauge Coherence: A Thermodynamic Synthesis of the Principle of Biological Time Equivalence}, arXiv:2607.04827 (2026).

\bibitem{Taye2026RelativisticPbteBiological}
M. A. Taye, \textit{Relativistic PBTE: Biological Proper Time Along the Worldline}, arXiv:2607.04849 (2026).

\bibitem{Taye2026NonequilibriumInternalTime}
M. Taye, \textit{A Nonequilibrium Internal-Time Model of Aging: Entropy-Normalized Biological Proper Time and Repair Bifurcations}, arXiv:2606.23279 (2026).

\bibitem{Taye2026BiologicalProperTime}
M. A. Taye, Int. J. Sci. Res. Publ. \textbf{16}, 17417 (2026).

\bibitem{Taye2026PrincipleBiologicalTime}
M. A. Taye, \textit{The Principle of Biological Time Equivalence:} (2026) [Metadata incomplete in Google Scholar export].

\bibitem{Taye2026NonequilibriumThermodynamicsStochastic}
M. A. Taye, \textit{Nonequilibrium Thermodynamics in Stochastic Processes} (2026) [Metadata incomplete in Google Scholar export].

\bibitem{Taye2026BrownianMotorsBrownian}
M. A. Taye, \textit{Brownian Motors and Brownian Heat Engines: From Classical Thermodynamics to Fluctuation-Driven Machines} (Independently published, 2026), 539 pp..

\bibitem{Taye2026ExactThermodynamicAnalysis}
M. A. Taye, \textit{Exact Thermodynamic Analysis of a Hybrid Molecular Motor Switching Between Active and Passive Modes} (2026) [Metadata incomplete in Google Scholar export].

\bibitem{Taye2026NoiseActivatedDopant}
M. A. Taye, \textit{Noise-Activated Dopant Dynamics in Two-Dimensional Thermal Landscapes with Localized Cold Spots} (2026) [International Journal of Scientific and Research Publications (IJSRP) 16 (5)].

\bibitem{Taye2026UniversalThermodynamicInequality}
M. A. Taye, \textit{A Universal Thermodynamic Inequality: Scaling Relations Between Current, Activity, and Entropy Production} (2026) [Metadata incomplete in Google Scholar export].

\bibitem{Taye2026ThermodynamicParametrisationVertebrate}
M. A. Taye, Int. J. Sci. Res. Publ. \textbf{16}, 2250 (2026).

\bibitem{Taye2026NeuralInvestmentEntropy}
M. Taye, \textit{Neural Investment as an Entropy-Budget Strategy: A Thermodynamic Derivation of Primate Longevity from the Principle of Biological Time Equivalence}, arXiv:2604.27937 (2026).

\bibitem{Taye2026LifetimeCardiacCycle}
M. Taye, \textit{The Lifetime Cardiac-Cycle Invariant in Endothermic Vertebrates: A 230-Species Comparative Dataset, Statistical Validation, and Explicit Falsifiability Criteria}, arXiv:2604.27856 (2026).

\bibitem{Taye2026BiologicalTimeEquivalence}
M. Taye, \textit{Biological Time Equivalence in Vertebrates: Thermodynamic Framework, Comparative Tests, and Clade-Specific Deviations}, arXiv:2603.26377 (2026).

\bibitem{Taye2026EntropyProductionMacroscopic}
M. A. Taye, Mod. Math. Phys. \textbf{2}, 1 (2026).

\bibitem{Taye2025CompetingActivePassive}
M. A. Taye, Physica A, 131214 (2025).

\bibitem{Taye2025CurzonAhlbornType}
M. A. Taye, Phys. Rev. E \textbf{112}, 044122 (2025).

\bibitem{Taye2025ThermodynamicIrreversibilityUnderdamped}
M. A. Taye, Phys. Rev. E \textbf{112}, 044101 (2025).

\bibitem{Taye2025UnifiedNonequilibriumFramework}
M. Taye, \textit{A Unified Nonequilibrium Framework: Thermodynamic Distance, Dissipation, and Stationary Laws via Effective State Count, Variational Stationarity, and Thermodynamic Bounds}, arXiv:2509.09041 (2025).

\bibitem{taye2025EntropyProductionThermodynamic}
M. A. Taye, Contemp. Math. \textbf{6}, 4101 (2025).

\bibitem{Taye2025ThermodynamicFeaturesHeat}
M. Taye, \textit{Thermodynamic Features of a Heat Engine Coupled with Exponentially Decreasing Temperature Across the Reaction Coordinate, as well as Perspectives on Nonequilibrium Thermodynamics}, arXiv:2503.24317 (2025).

\bibitem{Taye2025DrugWashoutViral}
M. A. Taye, bioRxiv, 2025.03. 22.644757 (2025).

\bibitem{Taye2025ThermodynamicRelationsTime}
M. Taye, arXiv e-prints, arXiv: 2503.20812 (2025).

\bibitem{Taye2025EntropyProductionThermodynamic}
M. A. Taye, \textit{Entropy Production and Thermodynamic Dynamics in Active and Passive Brownian Systems Driven by Time-Dependent Forces and Temperatures} (2025) [Metadata incomplete in Google Scholar export].

\bibitem{Taye2025DirectedTransportShort}
M. A. Taye, \textit{Directed Transport of a Short Polymer Chain on a Temperature-Dependent Ratchet Potential} (2025) [Metadata incomplete in Google Scholar export].

\bibitem{Taye2025ThermodynamicFeaturesHeat2}
M. A. Taye, \textit{Thermodynamic Features of a Heat Engine with an Exponentially Decreasing Temperature Profile} (2025) [Metadata incomplete in Google Scholar export].

\bibitem{Taye2024ExactTimeDependent}
M. A. Taye, Phys. Rev. E \textbf{110}, 054105 (2024).

\bibitem{Taye2024ExactTimeDependent2}
M. A. Taye, Contemp. Math., 5113-5149 (2024).

\bibitem{Taye2023TimeDependentThermodynamic}
M. A. Taye, bioRxiv, 2023.12. 06.570486 (2023).

\bibitem{Mahmud2023ComputationalInvestigationCis}
M. Mahmud, M. Bekele, and N. Behera, Theory Biosci. \textbf{142}, 151-165 (2023).

\bibitem{Ashagre2023SynergisticContributionPotassium}
S. Ashagre, A. K. Ogundele, J. N. Ike, B. Gebremichael, M. Bekele, G. D. Sharma, and ..., Journal of Physics and Chemistry of Solids 177, 111290 (2023).

\bibitem{Taye2023TimeDependentSolutions}
M. A. Taye, Eur. Phys. J. B \textbf{96}, 65 (2023).

\bibitem{Taye2023TransportingShortPolymer}
M. Taye, \textit{Transporting a short polymer along a reaction coordinate that coupled with a spatially varying temperature} (2023) [Metadata incomplete in Google Scholar export].

\bibitem{Taye2023DynamicsBloodCells}
M. A. Taye, bioRxiv, 2023.01. 21.525013 (2023).

\bibitem{Taye2023HostViralLoad}
M. A. Taye, Contemp. Math \textbf{4}, 392-410 (2023).

\bibitem{Mahmud2022ThermalQuantumFluctuation}
M. Mahmud, M. Bekele, and Y. Bassie, Condens. Matter \textbf{7}, 62 (2022).

\bibitem{Taye2022ExactTimeDependent}
M. Taye, \textit{Exact time-dependent analytical solutions for entropy production rate for a system that operates in a heat bath where its temperature varies linearly in space}, arXiv:2205.10322 (2022).

\bibitem{Taye2022ExactTimeDependent2}
M. A. Taye, Phys. Rev. E \textbf{105}, 054126 (2022).

\bibitem{Aragie2022NoiseFormedTriple}
B. Aragie, M. Bekele, and G. Pellicane, Pramana \textbf{96}, 59 (2022).

\bibitem{Abebe2022ThermallyActivatedDiffusion}
Y. Abebe, T. Birhanu, L. Demeyu, M. Taye, M. Bekele, and Y. Bassie, Eur. Phys. J. B \textbf{95}, 9 (2022).

\bibitem{Birhanu2021StochasticResonatorLayered}
T. Birhanu, Y. Abebe, L. Demeyu, M. Taye, and M. Bekele, Int. J. Mod. Phys. B \textbf{35}, 2150284 (2021).

\bibitem{Taye2021BrownianMotorsArranged}
M. A. Taye, Eur. Phys. J. B \textbf{94}, 124 (2021).

\bibitem{Taye2021EffectViscousFriction}
M. A. Taye, Phys. Rev. E \textbf{103}, 042132 (2021).

\bibitem{Zhang2021EfficacyAntiviralDrug}
S. Y. Zhang and M. A. Taye, \textit{The efficacy of antiviral drug, HIV viral load and the immune response}, arXiv:2101.10413 (2021).

\bibitem{Zhang2020HivViralLoad}
S. Y. Zhang and M. A. Taye, \textit{HIV viral load and the efficacy of antiviral drug} (2020) [Metadata incomplete in Google Scholar export].

\bibitem{Taye2020CorrelationBetweenAntiviral}
M. A. Taye, bioRxiv, 2020.11. 06.372094 (2020).

\bibitem{Taye2020SedimentationRateErythrocyte}
M. A. Taye, Eur. Phys. J. E \textbf{43}, 19 (2020).

\bibitem{Taye2020EntropyProductionEntropy}
M. A. Taye, Phys. Rev. E \textbf{101}, 012131 (2020).

\bibitem{Taye2019PhysicsErythrocyteSedimentation}
M. A. Taye, \textit{The physics of Erythrocyte Sedimentation Rate}, arXiv:1907.12148 (2019).

\bibitem{Duki2018StochasticResonanceFirst}
S. F. Duki and M. A. Taye, \textit{Stochastic resonance and first passage time for excitable system exposed to underdamped medium}, arXiv:1809.10017 (2018).

\bibitem{Duki2018StochasticResonanceFirst2}
S. F. Duki and M. A. Taye, J. Stat. Phys. \textbf{171}, 878-896 (2018).

\bibitem{Taye2017IrreversibleBrownianHeat}
M. A. Taye, J. Stat. Phys. \textbf{169}, 423-440 (2017).

\bibitem{Duki2016FirstPassageTime}
S. F. Duki and M. A. Taye, \textit{First passage time and stochastic resonance of excitable systems}, arXiv:1609.07752 (2016).

\bibitem{Taye2016FreeEnergyEntropy}
M. A. Taye, Phys. Rev. E \textbf{94}, 032111 (2016).

\bibitem{Taye2015EffectTemperatureDependence}
M. A. Taye and S. F. Duki, Eur. Phys. J. B \textbf{88}, 322 (2015).

\bibitem{Taye2015ExactAnalyticalThermodynamic}
M. A. Taye, Phys. Rev. E \textbf{92}, 032126 (2015).

\bibitem{Taye2015RectifiedMotionShort}
M. A. Taye, \textit{Rectified motion of short polymer chain that walks along a ratchet potential that coupled with spatially varying temperature}, arXiv:1507.04945 (2015).

\bibitem{Duki2015EffectTemperatureViscous}
S. F. Duki and M. A. Taye, \textit{The effect of temperature on viscous friction and the performance of a Brownian heat engine}, arXiv:1507.02005 (2015).

\bibitem{Taye2015ExactAnalyticalExpressions}
M. A. Taye, \textit{Exact analytical expressions for entropy production and free energy}, arXiv:1507.01791 (2015).

\bibitem{Aragie2014ImpurityDiffusionHarmonic}
B. Aragie, M. Asfaw, L. Demeyu, and M. Bekele, Eur. Phys. J. B \textbf{87}, 214 (2014).

\bibitem{Asfaw2014ThermallyActivatedBarrier}
M. Asfaw, \textit{Thermally activated barrier crossing rate for a coupled system moving in a ratchet potential}, arXiv:1407.3713 (2014).

\bibitem{Asfaw2014ThermodynamicFeatureBrownian}
M. Asfaw, Phys. Rev. E \textbf{89}, 012143 (2014).

\bibitem{Asfaw2013CorrectionTimingStatistics}
M. Asfaw, E. Alvarez-Lacalle, and Y. Shiferaw, PLoS ONE \textbf{8}, 10.1371/annotation/10d4ef64-c7e6-43ff-8bd7-658d47689855 (2013).

\bibitem{Asfaw2013TimingStatisticsSpontaneous}
M. Asfaw, E. Alvarez-Lacalle, and Y. Shiferaw, PLoS ONE \textbf{8}, e62967 (2013).

\bibitem{Asfaw2013EffectThermalInhomogeneity}
M. Asfaw, Eur. Phys. J. B \textbf{86}, 189 (2013).

\bibitem{Chen2012StatisticsCalciumMediated}
W. Chen, M. Asfaw, and Y. Shiferaw, Biophys. J. \textbf{102}, 461-471 (2012).

\bibitem{Shiferaw2012CalciumMediatedTriggered}
Y. Shiferaw, W. Chen, and M. Asfaw, Biophys. J. \textbf{102}, 672a (2012).

\bibitem{Asfaw2012ExploringDynamicsDimer}
M. Asfaw and Y. Shiferaw, \textit{Exploring the dynamics of dimer crossing over a Kramers type potential} (2012) [The Journal of Chemical Physics 136 (2)].

\bibitem{Asfaw2011ExploringElasticFeatures}
M. Asfaw, J. Phys.: Condens. Matter \textbf{23}, 105101 (2011).

\bibitem{Asfaw2011NoiseCreatedBistability}
M. Asfaw, B. Aragie, and M. Bekele, Eur. Phys. J. B \textbf{79}, 371-376 (2011).

\bibitem{Asfaw2011LateralPhaseSeparation}
M. Asfaw, Int. J. Mod. Phys. B \textbf{25}, 457-466 (2011).

\bibitem{Asfaw2010ThermallyActivatedBarrier}
M. Asfaw, Physical Review E—Statistical, Nonlinear, and Soft Matter Physics 82 (2 (2010).

\bibitem{Asfaw2010StochasticResonanceFlexible}
M. Asfaw and W. Sung, EPL \textbf{90}, 30008 (2010).

\bibitem{Bekele2009EvaluationClimateChange}
H. M. Bekele, \textit{Evaluation of climate change impact on upper blue Nile Basin reservoirs} (2009) [Arba Minch University 109].

\bibitem{Asfaw2008AdhesionInducedLateral}
M. Asfaw and H. Y. Chen, \textit{Adhesion-induced lateral phase separation of multi-component membranes: the effect of repellers and confinement}, arXiv:0811.3725 (2008).

\bibitem{Asfaw2008ModelingEfficientBrownian}
M. Asfaw, Eur. Phys. J. B \textbf{65}, 109-116 (2008).

\bibitem{Getahun2008CompetingJumpCycles}
Z. Getahun, M. Asfaw, and M. Bekele, \textit{Competing jump cycles for vacancy diffusion in binary alloys}, arXiv:0807.5034 (2008).

\bibitem{Asfaw2008UnbindingTransitionsMembranes}
M. Asfaw, Physica A \textbf{387}, 3526-3536 (2008).

\bibitem{Asfaw2007AdhesionInducedLateral}
M. Asfaw, Bussei Kenky\={u} \textbf{89}, 12 (2007).

\bibitem{Demeyu2007MonteCarloSimulations}
L. Demeyu, S. Stafström, and M. Bekele, Phys. Rev. B \textbf{76}, 155202 (2007).

\bibitem{Asfaw2007ExploringOperationTiny}
M. Asfaw and M. Bekele, Physica A \textbf{384}, 346-358 (2007).

\bibitem{Asfaw2007UnbindingTransitionsMulticomponent}
M. Asfaw, \textit{Unbinding transitions of multicomponent membranes and strings}, arXiv:0706.3433 (2007).

\bibitem{Asfaw2006MembraneAdhesionVia}
M. Asfaw, B. Różycki, R. Lipowsky, and T. R. Weikl, EPL \textbf{76}, 703-709 (2006).

\bibitem{Asfaw2005EnergeticsSimpleMicroscopic}
M. Asfaw and M. Bekele, Physical Review E—Statistical, Nonlinear, and Soft Matter Physics 72 (5 (2005).

\bibitem{Asfaw2005AdhesionMultiComponent}
M. Asfaw, \textit{Adhesion of multi-component membranes and strings} (2005).

\bibitem{Asfaw2004CurrentMaximumPower}
M. Asfaw and M. Bekele, \textit{Current, maximum power and optimized efficiency of a Brownian heat engine} (2004) [Metadata incomplete in Google Scholar export].

\bibitem{Pearson2003ManagementHealthReproduction}
R. A. Pearson, M. Alemayehu, A. Tesfaye, D. G. Smith, G. Kebede, and M. Asfaw, \textit{Management, health and reproduction of donkeys used for work in peri-urban areas of West and East Shewa, Ethiopia: a survey} (2003).

\bibitem{PEARSON2003UseManagementDonkeys}
R. A. PEARSON, M. ALEMAYEHU, A. TESFAYE, E. F. ALLAN, D. G. SMITH, and ..., \textit{Use and management of donkeys in peri-urban areas of Ethiopia. Phase One. 133 pp.} (2003) [Metadata incomplete in Google Scholar export].

\bibitem{Asfaw2002AdjustableBrownianHeat}
M. Asfaw and M. Bekele, \textit{An adjustable Brownian heat engine}, arXiv:preprint (2002).

\end{thebibliography}
\end{document}